\documentclass[a4paper, UKenglish, cleveref, autoref, thm-restate]{lipics-v2021}

\usepackage{settigs}
\usepackage{macros}

\title{Notes on $k\aexppol$ problems\\ for deterministic machines} 

\titlerunning{Notes on $k\aexppol$ problems for deterministic machines} 

\author{Alessio Mansutti}{Department of Computer Science, University of Oxford, Oxford, UK}{alessio.mansutti@cs.ox.ac.uk}{https://orcid.org/0000-0002-1104-7299}{}

\authorrunning{A.~Mansutti} 

\Copyright{Alessio Mansutti} 



\category{} 

\relatedversion{} 

\acknowledgements{}

\nolinenumbers 

\EventEditors{}
\EventNoEds{1}
\EventLongTitle{}
\EventShortTitle{}
\EventAcronym{}
\EventYear{}
\EventDate{}
\EventLocation{}
\EventLogo{}
\SeriesVolume{0}
\ArticleNo{0}

\begin{document}

\maketitle

\begin{abstract}
    The complexity of several logics, such as Presburger arithmetic, dependence logics and ambient logics, can only be characterised in terms of alternating Turing machines. Despite quite natural, the presence of alternation can sometimes cause neat ideas to be obfuscated inside heavy technical machinery. In these notes, we propose two problems on deterministic machines that can be used to prove lower bounds with respect to the computational class $k\aexppol$, i.e.~the class of all problems solvable by an alternating Turing machine running in $k$ exponential time and performing a polynomial amount of alternations, with respect to the input size. The first problem, called \emph{$k\aexppol$-prenex \TM problem}, is a problem about deterministic Turing machines. The second problem, called the \emph{$k$-exp alternating multi-tiling problem}, is analogous to the first one, but on tiling systems.

Both problems are natural extensions of the \emph{TM alternation problem} and the \emph{alternating multi-tiling problem} proved $\aexppol$-complete by L.~Bozzelli, A.~Molinari, A.~Montanari and A.~Peron in~\cite{Bozzelli17}. The proofs presented in these notes follow the elegant exposition in A.~Molinari's PhD thesis~\cite{Molinari19} to extend these results from the case $k = 1$ to the case of arbitrary~$k$.

\end{abstract}

\section{Alternation problems for deterministic Turing machines}
\label{sec:alternationproblems}

In this section, we introduce the \defstyle{$k\aexppol$-prenex \TM problem}, a decision problem for deterministic multi-tape Turing machines,
which we prove being $k\aexppol$-complete in~\Cref{appendix:ATMs} (see below for the definition of this complexity class). 
The $k\aexppol$-prenex problem is a straightforward generalisation of the \defstyle{TM alternation problem} shown $\aexppol$-complete in~\cite{Bozzelli17} and in~\cite[page~292]{Molinari19}, from which our (self-contained) proofs are based upon.

\subparagraph*{Notation for regular languages.} 
Let~$A$ and $B$ be two regular languages.
As usual, we write $A \cup B$ and $A \setminus B$ for the set theoretical union and difference of languges, respectively.
With~$A \cdot B$ we denote the language obtained by concatenating words in~$A$ with words in~$B$.
Then,~$A^0 = \{\epsilon\}$ and for all $n \in \Nat$, $A^{n+1} = A^n \cdot A$. Lastly, $A^* = \bigcup_{n \in \Nat} A^n$ and  $A^+ \egdef A^* \setminus A^0$.
Given a finite~\emph{alphabet} $\Alphabet$, 
the set $\Alphabet^{n}$ is the set of words over~$\Alphabet$ of length~$n \in \Nat$ and,
given a finite word~$\aword$, we write $\card{\aword}$ for its length.

\subparagraph*{Complexity classes.}
We recall some standard complexity classes that appear throughout these notes. 
Below and across the whole paper, given natural numbers $k,n \geq 1$, we write $\tetra$ for the tetration function inductively defined as $\tetra(0,n) \egdef n$ and ${\tetra(k, n) = 2^{\tetra(k-1,n)}}$. Intuitively, $\tetra(k,n)$ defines a tower of exponentials of height~$k$.
\begin{itemize}[nosep]
  \item $k\nexptime$ is the class of all problems 
decidable with a non-deterministic Turing machine
running in time~$\tetra(k,f(n))$ for some polynomial~$f : \Nat \to \Nat$, on each input of length $n$.
  \item $\Sigma^{k\exptime}_j$ is the class of all problems decidable
  with an alternating Turing machine (\ATM,~\cite{ChandraKS81}) in time $\tetra(k,f(n))$, starting on an existential state and performing at most $j-1$ alternations, for some polynomial $f : \Nat \to \Nat$, on each input of length~$n$. By definition, $\Sigma^{k\exptime}_1 = k\nexptime$.
  \item $k\aexppol$ is the class of all problems decidable with an \ATM running
  in time~$\tetra(k,f(n))$ and performing at most $g(n)$ number of alternations, for some polynomials~$f,g$, on each input of length~$n$. The inclusion $\Sigma^{k\exptime}_j \subseteq k\aexppol$ holds for every $j \in \PNat$.
\end{itemize}

\subparagraph*{Deterministic $n$-tapes TM.}
Let $n \in \PNat$.
Following the presentation in~\cite{Molinari19}, 
we define the class of \textit{deterministic $n$-tapes Turing machines} (\nTM).
A~\nTM is 
a deterministic machine $\Turing = (n,\Alphabet,\States,\qinit,\qacc,\qrej,\trans)$,
where $\Sigma$ is a finite~alphabet containing a \emph{blank symbol}~$\blanksymb$, and 
$\States$~is a finite set of \emph{states} including the initial state~$\qinit$, the accepting state~$\qacc$ and the rejecting state~$\qrej$.
The \nTM operates on $n$ distinct \emph{tapes} numbered from $1$ to $n$, and it has one \emph{read/write head} shared by all tapes. Every \emph{cell} of each tape contains a symbol from $\Alphabet$ and is indexed with a number from $\Nat$ (i.e.~the tapes are infinite only in one direction).
Lastly, the deterministic transition function~${\trans\colon \States \times \Sigma \to (\States \times \Sigma \times \{-1,+1\}) \cup (\States \times [1,n])}$ is such that for every $\asymbol \in \Alphabet$, $\trans(\qacc,a) \egdef (\qacc,1)$ and $\trans(\qrej,a) \egdef (\qrej,1)$.

A configuration of~$\Turing$ is given by the content of the $n$ tapes together with a \emph{positional state}~$(\astate,j,k) \in \States \times [1,n] \times \Nat$.
The content of the tape is represented by an $n$-uple of finite words $(\aword_1,\dots,\aword_n) \in (\Alphabet^*)^n$, where each 
$\aword_i$ specifies the content of the first $\card{\aword_i}$ cells of the $i$-th tape, after which the tape only contains the blank symbols~$\blanksymb$.
The positional state~$(\astate,j,k)$
describes the current state~$\astate$ of the machine together with the position~$(j,k)$ of the read/write head, placed on the $k$-th cell of the $j$-th tape.
At each step, given the positional state~$(\astate,j,k)$ of $\Turing$ and the symbol~$\asymbol \in \Alphabet$ read by the read/write head in position~$(j,k)$, one of the following occurs:
\begin{itemize}[leftmargin=*,align=left]
    \item[ \textit{case $\trans(\astate,\asymbol) = (\astate',\asymbolbis,i)$, where $\astate \in \States$, $\asymbolbis \in \Alphabet$ and $i \in \{-1,1\}$ (ordinary moves)}:] 
    If $k + i \in \Nat$,
    then $\Turing$ overwrites the symbol $\asymbol$ in position $(j,k)$ with the symbol $\asymbolbis$. Afterwards~$\Turing$ changes its positional state to $(\astate',j,k+i)$. 
    Otherwise ($k + i \notin \Nat$, i.e.~$k = 0$ and $i = -1$),
    $\Turing$ does not modify the tapes and changes its positional state to $(\qrej,1,0)$.
    \item[\textit{case $\trans(\astate,\asymbol) = (\astate',j')$, where $\astate' \!\in\, \States$ and $j' \!\in\, [1,n]$ (jump moves)}:] 
    The read/write head moves to the $k$-th cell of the $j'$-th tape, i.e.~$\Turing$ changes its positional state to $(\astate',j',k)$.
    Jump moves do not modify the content of the $n$ tapes.
\end{itemize}

Initially, each tape contains a word in $\Sigma^*$ written from left to right starting from the position~$0$ of the tape.
Hence, an input of~$\Turing$ can be described as a tuple $(\aword_1,\dots,\aword_n) \in (\Sigma^*)^n$ where, for all $j \in [1,n]$, 
$\aworld_j$ is the input word for the $j$-th tape.
We write $\Turing(\aworld_1,\dots,\aworld_n)$ for the run of the \nTM on input $(\aworld_1,\dots,\aworld_n) \in (\Sigma^*)^n$, 
starting from the positional state $(\qinit,1,0)$. 
The run~$\Turing(\aworld_1,\dots,\aworld_n)$~\emph{accepts} in time~$t \in \Nat$ if $\Turing$ reaches the accepting state~$\qacc$ in at most~$t$ steps.

\begin{definition}[$k\aexppol$-prenex \TM problem]
    \label{definition:KEXP-alt-problem}
    Fix~$k \in \PNat$.

    \begin{tabular}{rl}
    {\rm\textbf{Input:}}&$(n,\vec{Q},\Turing)$, where~${n \in \PNat}$ is written in unary, ${\vec{Q} = (Q_1,\dots,Q_n) \in \{\exists,\forall\}^n}$ \\ &with $Q_1 = \exists$, and a $\Turing$ is a \nTM~working on alphabet~$\Alphabet$.\\
    {\rm\textbf{Question:}}&is it true that\\
        &$Q_1 \aworld_1 \in \Sigma^{\tetra(k,n)}, \dots, Q_n \aworld_n \in \Sigma^{\tetra(k,n)}$ : 
        $\Turing(\aworld_1,\dots,\aworld_n)$ accepts in time
        $\tetra(k,n)$ ?
    \end{tabular}
    
\end{definition}

We analyse the complexity of the $k\aexppol$-prenex \TM problem depending on the type of \emph{quantifier prefix}~$\vec{Q}$. For arbitrary quantifier prefixes, 
the problem is $k\aexppol$-complete.
The membership in~$k\aexppol$ is straightforward, whereas 
the hardness follows by reduction from the 
acceptance problem for alternating Turing machines running in  $k$-exponential time and performing a polynomial number of alternations, as we show in~\Cref{appendix:ATMs}.

\begin{restatable}{theorem}{TheoremKAEXPAltProblem}
    \label{theorem:kaexp-alternation-complexity}
    The $k\aexppol$-prenex \TM problem is $k\aexppol$-complete.
\end{restatable}

For bounded alternation,
let $\altern{\vec{Q}}$ be the number of alternations between existential and universal quantifiers in~$\vec{Q} \in \{\exists,\forall\}^n$, plus one. That is, $\altern{\vec{Q}} \egdef 1 + \card{\{i \in [2,n] : Q_i \neq Q_{i-1}\}}$. For $j \in \PNat$, the $\Sigma^{k\exptime}_j$-prenex \TM problem is the $k\aexppol$-prenex \TM problem restricted to instances~$(n,\vec{Q},\Turing)$ with $\altern{\vec{Q}} = j$.
Refining~\Cref{theorem:kaexp-alternation-complexity}, this problem is~$\Sigma^{k\exptime}_j$-complete.

\begin{restatable}{theorem}{CorollarySIGMAKProblem}
    \label{theorem:kaexp-j-prenex-complexity}
    The $\Sigma^{k\exptime}_j$-prenex \TM problem is $\Sigma^{k\exptime}_j$-complete.
\end{restatable}

\Cref{theorem:kaexp-j-prenex-complexity} implies that the problem is $k\nexptime$-complete on instances where $\vec{Q} \in \{\exists\}^n$.

\section{Alternation problems for multi-tiling systems}

The $k\aexppol$-prenex TM problem can be easily recast as a problem on multi-tiling systems.

\subparagraph*{Multi-tiling.}
A \emph{multi-tiling system} $\cP$ is a tuple $(\cT,\cT_0,\cTacc,\cH,\cV,\cM,n)$
such that $\cT$ is a finite set of \defstyle{tile types},
$\cT_0,\cTacc \subseteq \cT$ are sets of \defstyle{initial} and \defstyle{accepting} tiles, respectively, $n \in \PNat$ (written in unary) is the \emph{dimension} of the system, and $\cH,\cV,\cM \subseteq \cT \times \cT$ represent the
horizontal, vertical and multi-tiling matching relations, respectively.

Fix $k \in \PNat$.
We write $\widehat{\Alphabet}$ for the set of words of length $\tetra(k,n)$ over an alphabet~$\Alphabet$.
The \defstyle{initial row} $\tilinginit{\amap}$
of a map $\amap \colon [0,\tetra(k,n)-1]^2 \to \cT$
is the word $\amap(0,0),\amap(0,1),\dots,\amap(0,\tetra(k,n){-}1)$ from $\widehat{\cT}$.
A \emph{tiling} for the \emph{grid} $[0,\tetra(k,n)-1]^2$ is a tuple $(\amap_1,\amap_2,\dots,\amap_n)$ such that
\begin{description}[nosep,itemsep=1pt,topsep=1.5pt,after=\vspace{1.5pt}]
    \item[maps.\namedlabel{multit-maps}{\textbf{maps}}]
    $f_\ell \colon [0,\tetra(k,n)-1]^2 \to \cT$ assigns a tile type to each position
    of the grid, for all $\ell \in [1,n]$;
    \item[init \& acc.\namedlabel{multit-initacc}{\textbf{init \& acc}}]
    $\tilinginit{f_\ell} \in \widehat{\cT}_0$ for all $\ell \in [1,n]$, and 
    $\amap_n(\tetra(k,n)-1,j) \in \cTacc$ for some $0 \leq j < \tetra(k,n)$;
    \item[hori.\namedlabel{multit-hori}{\textbf{hori}}] $(\amap_\ell(i,j),\amap_\ell(i+1,j)) \in \cH$, 
    for every $\ell \in [1,n]$, $i \in [0,\tetra(k,n)-2]$ and $0 \leq j < \tetra(k,n)$;
    \item[vert.\namedlabel{multit-vert}{\textbf{vert}}] 
    $(\amap_\ell(i,j),\amap_\ell(i,j+1)) \in \cV$,
    for every $\ell \in [1,n]$, $j \in [0,\tetra(k,n)-2]$ and $0 \leq i < \tetra(k,n)$;
    \item[multi.\namedlabel{multit-multi}{\textbf{multi}}] 
    $(\amap_{\ell}(i,j),\amap_{\ell+1}(i,j)) \in \cM$,
    for every $1 \leq \ell < n$ and $0 \leq i,j < \tetra(k,n)$.
\end{description}
\vspace{2pt}

\begin{definition}[$k$-exp alternating multi-tiling problem]
    \label{definition:KEXP-alt-problem}
    Fix~$k \in \PNat$.

    \begin{tabular}{rl}
    {\rm\textbf{Input:}} & a multi-tiling system~$\cP = (\cT,\cT_0,\cTacc,\cH,\cV,\cM,n)$\\ & and ${\vec{Q} = (Q_1,\dots,Q_n) \in \{\exists,\forall\}^n}$ with $Q_1 = \exists$.\\
    {\rm\textbf{Question:}} & is it true that\\
        &\begin{tabular}{rl}
            ${Q_1 \aword_1 \in \widehat{\cT_0}} \dots {Q_n \aword_n \in \widehat{\cT_0}}$ : &
            there is a tiling $(\amap_1,\dots,\amap_n)$ of $[0,\tetra(k,n)-1]^2$\\ 
            &such that~$\tilinginit{\amap_\ell} = \aword_\ell$ for all $\ell \in \interval{1}{n}$ ?
        \end{tabular}
    \end{tabular}
\end{definition}

For the case where $k = 1$, deciding whether an instance of the $k$-exp alternating multi-tiling problem is accepted 
has been shown $\aexppol$-complete in~\cite[App.~E.7]{Molinari19}, whereas $k\aexppol$-completeness for arbitrary~$k$ can be shown thanks to~\Cref{theorem:kaexp-alternation-complexity}. Moreover, the problem becomes~$\Sigma^{k\exptime}_j$-complete whenever $\altern{\vec{Q}} = j$, by~\Cref{theorem:kaexp-j-prenex-complexity}.

\begin{restatable}{theorem}{TheoremKTILINGProblem}
    \label{theorem:kaexp-TILING-complexity}
    The $k$-exp alternating multi-tiling problem is $k\aexppol$-complete.
\end{restatable}

\begin{restatable}{theorem}{CorollarySIGMAKTILINGProblem}
    \label{theorem:kaexp-j-prenex-TILING-complexity}
    The $k$-exp alternating multi-tiling problem is $\Sigma^{k\exptime}_j$-complete when restricted to inputs~$(\cP, \vec Q)$ such that $\altern{\vec{Q}} = j$.
\end{restatable}

With~\Cref{theorem:kaexp-alternation-complexity,theorem:kaexp-j-prenex-complexity} and results in~\cite[App.~E.7]{Molinari19} at hand, showing these two theorems is very simple.
Indeed, standard conversions between Turing machines and tiling systems works in this case, and in particular the proof of~\cite[Prop.~E.7.6]{Molinari19} shows
a well-known conversion between time-space diagrams of computations on TM and tiling systems where the set of domino types and matching relations is entirely determined by the description of the TM, and thus not by its runtime. 
This conversion is enough to establish the lower bounds of~\Cref{theorem:kaexp-alternation-complexity,theorem:kaexp-j-prenex-complexity} (the upper bounds are instead straightforward).
Because of this, and to keep the notes short, 
in the remaining sections we simply focus on establishing~\Cref{theorem:kaexp-alternation-complexity,theorem:kaexp-j-prenex-complexity}.

\section{Some properties of~$\tetra(k,n)$.}
Before moving to the proofs of~\Cref{theorem:kaexp-alternation-complexity} and~\Cref{theorem:kaexp-j-prenex-complexity}, 
we discuss some trivial properties of the tetration function~$\tetra$, 
which we later need.
In~\Cref{lemma:tetra-property-1,lemma:tetra-property-2,lemma:tetra-property-3} below, let~$k,n \in \PNat$.

\begin{lemma}
    \label{lemma:tetra-property-1}
    $p(\tetra(k,n)) \leq \tetra(k,p(n))$, 
    for every polynomial~$p(\avar) = \alpha \avar^d + \beta$ with~$\alpha,d,\beta \in \PNat$.
\end{lemma}

\begin{proof}
    Straightforward induction on~$k \in \PNat$:

    \begin{description}[nosep,leftmargin=*,itemsep=3pt]
        \item[base case: $k = 1$.] 
         Trivially, $\alpha (2^n)^d + \beta \leq 2^{\alpha n^d + \beta}$.
        \item[induction step: $k > 1$.]
        By induction hypothesis, $\alpha\tetra(k-1,n)^d + \beta \leq \tetra(k-1, \alpha n^d + \beta)$, so,\\[3pt]
        $\alpha \tetra(k,n)^d + \beta 
        = \alpha (2^{\tetra(k-1,n)})^d + \beta 
        \leq 2^{\alpha\tetra(k-1,n)^d + \beta}
        \leq 2^{\tetra(k-1, \alpha n^d + \beta)}
        = \tetra(k,\alpha n^d + \beta)$.
        \qedhere
    \end{description}
\end{proof}

\begin{lemma}
    \label{lemma:tetra-property-2}
    $\tetra(k,n)^2 \leq \tetra(k,2n)$.
\end{lemma}

\begin{proof}
    For $k = 1$, we have $\tetra(1,n) = (2^n)^2 = 2^{2n} = \tetra(1,2n)$. For $k > 1$, by~\Cref{lemma:tetra-property-1} $2\tetra(k-1,n) \leq \tetra(k-1,2n)$. So,
    ${\tetra(k,n)^2 = (2^{\tetra(k-1,n)})^2 = 2^{2\tetra(k-1,n)} \leq 2^{\tetra(k-1,2n)} = \tetra(k,2n)}$.
\end{proof}

We inductively define the function $\klog{k} : \Nat \to \Nat$ as follows:
\begin{center}
$\klog{1}(\avar) \egdef \ceil{\log_2(\avar)}$ \quad
and \quad
$\klog{k+1}(\avar) \egdef \klog{k}(\ceil{\log_2(\avar)})$.
\end{center}

\begin{lemma}
    \label{lemma:tetra-property-3}
    For every $m \in \Nat$, \
    $m > \tetra(k,n)$ if and only if $\klog{k}(m) > n$.
\end{lemma}

\begin{proof}
    Straightforward induction on~$k \in \PNat$:

    \begin{description}[nosep,leftmargin=*,itemsep=3pt]
        \item[base case: $k = 1$.] 
         \ProofRightarrow
         Suppose $m > 2^n$. Then, $\klog{1}(m) = \ceil{\log_2(m)} \geq \log_2(m) > n$.

         \ProofLeftarrow 
         Conversely, suppose $m \leq 2^n$. Then $\log_2(m) \leq n$. As $n \in \Nat$, $\ceil{\log_2(m)} \leq n$.
        \item[induction step: $k > 1$.]
        \ProofRightarrow 
        Suppose $m > \tetra(k,n)$. So, $\ceil{\log_2(m)} \geq \log_2(m) > \tetra(k-1,n)$.
        By induction hypothesis, $\klog{k-1}(\ceil{\log_2(m)}) > n$, i.e.~$\klog{k}(m) > n$.

        \ProofLeftarrow 
        Conversely, suppose $m \leq \tetra(k,n)$. 
        Therefore, $\log_2(m) \leq \tetra(k-1,n)$. 
        Since $n \in \Nat$, 
        $\ceil{\log_2(m)} \leq \tetra(k-1,n)$.
        By induction hypothesis, $\klog{k}(m) = \klog{k-1}(\ceil{\log_2(m)}) \leq n$.
        \qedhere
    \end{description}
\end{proof}

\Cref{lemma:tetra-property-3} allows us to check whether $m > \tetra(k,n)$ holds in time~$\BigO{k \cdot m + n}$,
by first computing~$r = \klog{k}(m)$ and then testing whether $r > n$.
Fundamentally, in this way we avoid computing~$\tetra(k,n)$ explicitly.
This fact is later used 
in~\Cref{lemma:almost-theorem-kaexp-alternation}, where we introduce a \nTM that shall test whether its inputs~$\aword_1,\dots,\aword_n$ are greater than~$\tetra(k,n)$, while running in polynomial time in $\card{\aword_1},\dots,\card{\aword_n}$, $k$ and $n$.
For similar reasons, this machine also requires to compute the 
\defstyle{integer square root}, defined as~$\isqrt{\avar} \egdef \floor{\sqrt{\avar}\,}$.

\begin{lemma}
    \label{lemma:isqrt}
    Given $n, m \in \Nat$, \
    $m \geq n^2$ if and only if $\isqrt{m} \geq n$.
\end{lemma}

\begin{proof}
    \ProofRightarrow 
    Suppose $m \geq n^2$. Then, $\sqrt{m} \geq n$. 
    Since $n \in \Nat$, we conclude $\floor{\sqrt{m}} \geq n$.

    \ProofLeftarrow 
    Suppose $\isqrt{m} \geq n$. Trivially, $\sqrt{m} \geq n$ and so $m \geq n^2$.
\end{proof}

\section{Alternating Turing machines and proof of~\Cref{theorem:kaexp-alternation-complexity}}
\label{appendix:ATMs}

This section contains the auxiliary technical tools needed to prove~\Cref{theorem:kaexp-alternation-complexity}, as well as the proof of this theorem and 
of~\Cref{theorem:kaexp-j-prenex-complexity}.
The section does not assume any prior knowledge on alternating Turing machines (ATM), which we define adapting the presentation of~\cite{ChandraKS81}.

\subparagraph*{Alternating Turing machines.}
An~ATM 
$\ATuring = (\Alphabet,\States_{\exists},\States_{\forall},\qinit,\qacc,\qrej,\trans)$
is a single-tape machine 
where $\Alphabet$ is a finite \defstyle{alphabet} containing at least two symbols, one of which is the \defstyle{blank symbol}~$\blanksymb$, 
the sets
$\States_{\exists}$ and $\States_{\forall}$ are two disjoint sets of \defstyle{states}, respectively called \defstyle{existential} and \defstyle{universal} states, $\qinit \in \States_{\exists}$ is the \defstyle{initial state}, and $\qacc$ and $\qrej$ are two auxiliary states 
that do not belong to $\States_{\exists} \cup \States_{\forall}$. 
The state~$\qacc$ is the \defstyle{accepting state}
and the state~$\qrej$ is the \defstyle{rejecting state}.
Every \defstyle{cell} of the tape contains a symbol from $\Alphabet$ and is indexed with a number from $\Nat$ (i.e.~the tapes are infinite only in one direction).
Let $\States \egdef \States_{\exists} \cup \States_{\forall} \cup \{\qacc,\qrej\}$.
The \defstyle{transition function} is a function of the form~$\trans \colon \States \times \Alphabet \to \powerset{(\States \times \Alphabet \times \{-1,+1\})}$, 
where for every~$\asymbol \in \Alphabet$, $\delta(\qacc,\asymbol) = \{(\qacc,\asymbol,+1)\}$,  $\delta(\qrej,\asymbol) = \{(\qrej,\asymbol,-1)\}$, 
and for every $\astate \in \States_{\forall} \cup \States_{\exists}$, $\delta(\astate,\asymbol) \neq \emptyset$.
We write $\card{\ATuring}$ for the size of~$\ATuring$, i.e.~the number of symbols needed in order to describe~$\ATuring$.

A \defstyle{configuration} of~$\ATuring$ is given by a triple $(\aword,\astate,k)$. 
The finite word~$\aword \in \Alphabet^*(\Alphabet \setminus \{\blanksymb\})$ specifies the content of the first $\card{\aword}$ cells fo the tape, after which the tape only contains the blank symbol~$\blanksymb$. Notice that $\aword$ does not end with~$\blanksymb$, which means that distinct words do not describe the same content of the tape.
The~\emph{positional state}~$(\astate,k)$ describes the current state~$\astate$ of the machine together with the position of the read/write head on the tape, here corresponding to the $k$-th cell. 

Let~$c = (\aword,\astate,k)$ be a configuration of~$\ATuring$ and consider the symbol $\asymbol \in \Alphabet$ read by the read/write head.
We write $\Delta(c)$ for the set of configurations reachable in exactly one step from~$c$. 
In particular, $(\aworld',\astate',k') \in \Delta(c)$ if and only if
there is $(\astate'',\asymbolbis,i) \in \trans(\astate,a)$ such that
\begin{itemize}
    \item $k+i \in \Nat$ (i.e.~$k \neq 0$ or $i \neq -1$), $k' = k+i$, 
    $\astate'' = \astate'$ and $\aworld'$ describes the tape after the read/write head modifies the content of the $k$-th cell to~$\asymbolbis$, or 
    \item $k + i \not\in \Nat$, $k' = 0$, $\aworld' = \aworld$ and $\astate' = \qrej$.
\end{itemize}
A~\emph{computation path} of~$\ATuring$ is a sequence $(c_0,\dots,c_d)$ of configurations such that $c_{i} \in \Delta(c_{i-1})$ holds for every $i \in [1,d]$. If the states~$\qacc$ and~$\qrej$ appear in the last configuration~$c_d$,
the path is called \emph{terminating}.
Notice that if $\qacc$ (or~$\qrej$) belongs to some $c_i$ ($i \in [1,d]$), then by definition of $\trans$ it also belongs to every $c_j$ with $j > i$, and thus the configuration is terminating.

To describe the notion of acceptance for~$\ATuring$,
we introduce a \emph{partial} labelling function $\gamma_{\ATuring} \colon \Alphabet^* \times \States \times \Nat \to \{\tacc,\trej\}$ 
(we drop the prefix~$\ATuring$ from~$\gamma_{\ATuring}$ when clear from the context).
Given $\aword \in \Alphabet^*$ and $k \in \Nat$, 
the function~$\gamma$ is defined as follows:
\begin{itemize}
    \item $\gamma(\aword,\qacc,k) \egdef \tacc$ and $\gamma(\aword,\qrej,k) \egdef \trej$,
    \item for every $\astate \in \States_{\exists}$,
    \begin{itemize}
        \item $\gamma(\aword,\astate,k) = \tacc$ if 
        $\gamma(\aword',\astate',k') = \tacc$ holds for some~$(\aworld',\astate',k') \in \Delta(\aword,\astate,k)$,
        \item $\gamma(\aword,\astate,k) = \trej$ if 
        $\gamma(\aword',\astate',k') = \trej$ holds for every~$(\aworld',\astate',k') \in \Delta(\aword,\astate,k)$,
        \item otherwise, $\gamma(\aword,\astate,k)$ undefined,
    \end{itemize}
    \item for every $\astate \in \States_{\forall}$,
    \begin{itemize}
        \item $\gamma(\aword,\astate,k) = \tacc$ if 
        $\gamma(\aword',\astate',k') = \tacc$ holds for every~$(\aworld',\astate',k') \in \Delta(\aword,\astate,k)$,
        \item $\gamma(\aword,\astate,k) = \trej$ if 
        $\gamma(\aword',\astate',k') = \trej$ holds for some~$(\aworld',\astate',k') \in \Delta(\aword,\astate,k)$,
        \item otherwise $\gamma(\aword,\astate,k)$ undefined.
    \end{itemize}
\end{itemize}
The language described by $\ATuring$ is $\alang(\ATuring) \egdef \{\aword \in \Alphabet^* \mid \gamma(\aword,\qinit,0) = \tacc\}$.
For our purposes, we are only interested in the notions of time-bounded and alternation-bounded acceptance.
We say that~$\ATuring$ \emph{accepts} (resp.~\emph{rejects}) a word~$\aword \in \Alphabet^*$ in time~$t \in \Nat$ whenever examining all the terminating computation paths~$(c_0,\dots,c_d)$, where $c_0 = (\aword,\qinit,0)$ and $d \in [0,t]$, is sufficient to conclude whether $\gamma(\aword,\qinit,0) = \tacc$ (resp.~$\gamma(\aword,\qinit,0) = \trej$). 
The machine~$\ATuring$ \emph{halts} on the word~$\aword$ in time~$t$
if it either accepts or rejects~$\aword$ in time~$t$.

Consider functions~$f,g \colon \Nat \to \Nat$.
The \ATM~$\ATuring$ is $f$-time bounded and~$g$-alternation bounded if and only if for every $\aword \in \Alphabet^*$, $\ATuring$ halts on $\aword$ in time~$f(\card{\aword})$ and, 
along each terminating computation path~$(c_0,\dots,c_d)$, where $c_0 = (\aword,\qinit,0)$ and~$d \in [0,t]$, the positional state of the machine alternate between existential and universal states at most~$g(\card{\aword})$ times. 

\begin{proposition}[\cite{ChandraKS81}]
    \label{proposition:kaexppol-membership}
    Consider $k \in \Nat_+$, polynomials~$f,g \colon \Nat \to \Nat$, and~$h(\avar) \egdef \tetra(k,f(\avar))$. Let $\ATuring$ be an $h$-time bounded and $g$-alternation bounded \ATM on alphabet~$\Alphabet$. Let~${\aword \in \Alphabet^*}$.
    The problem of deciding whether $\aword \in \alang(\ATuring)$
    is $k\aexppol$-complete.
\end{proposition}

\subparagraph*{From $\ATM$ to $\nTM$.}
Working towards a proof of~\Cref{theorem:kaexp-alternation-complexity}, we now aim at defining a~\nTM that checks if its input words represent computation paths of an alternating Turing machine.
Throughout this section, we fix~$k \in \Nat_+$ (encoded in unary), two polynomials~$f,g \colon \Nat \to \Nat$ and~$h(\avar) \egdef \tetra(k,f(\avar))$.
We consider a~$h$-time bounded and $g$-alternation bounded \ATM 
$\ATuring = (\Alphabet,\States_{\exists},\States_{\forall},\qinit,\qacc,\qrej,\trans)$.
Let $\aword \in \Alphabet^*$.
Since $\ATuring$ is \mbox{$h$-time} bounded, we can represent the configurations of~$\ATuring$ as words of the finite language $C = \{ \widehat{\aword} \in \Alphabet^* \cdot \States \cdot \Alphabet^+ \mid \card{\widehat{\aword}} =  h(\card{\aword}) \}$.
More precisely, we say that a word $\widehat{\aword} \in C$ \emph{encodes the configuration} $(\aword',\astate,k) \in \Alphabet^* \times \States \times \Nat$ 
whenever there are words~$\aword'' \in \Alphabet^*$, $\aword''' \in \Alphabet^+$
and $\aword_{\blanksymb} \in \{\blanksymb\}^*$. 
such that $\aword' \cdot \aword_{\blanksymb} = \aword'' \cdot \aword'''$, $\card{\aword''} = k$ and $\widehat{\aword} = \aword'' \cdot \astate \cdot \aword'''$.
For instance, the configuration $(\asymbol\asymbolbis, \astate, 3)$ is encoded by the word $\asymbol\asymbolbis\blanksymb\astate\blanksymb\dots\blanksymb$ of length $h(\card{\aword})$.
Each word~$\widetilde{\aword} \in C$ encodes exactly one configuration of~$\ATuring$, which we denote by~$c(\widetilde{\aword})$.

We extend our encoding to computation paths of~$\ATuring$, 
with the aim at defining a~\nTM that checks if its input words represent a computation path of~$\ATuring$. To this end, we introduce a new symbol~$\trail$ that does not appear in~$\Alphabet \cup \States$ and that is used to delimit the end of a computation path.
Let $\overline{\Alphabet} = \Alphabet \cup \States \cup \{\trail\}$  and
consider a word~$\overline{\aword} \in \overline{\Alphabet}^*$.
Again, since $\ATuring$ is $h$-time bounded, its computation paths can be encoded by a sequence of at most $h(\card{\aword})$ words from $C$ (so,~a word of length at most~$h(\card{\aword})^2$). 
Hence, let~$\widetilde{\aword} \in (\Alphabet \cup \States)^*$ be the only prefix of the word $\overline{\aword} \cdot \trail$ such that either $\widetilde{\aword} \cdot \trail$ is a prefix of $\overline{\aword} \cdot \trail$ or $\card{\widetilde{\aword}} = h(\card{\aword})^2$.
We say that $\overline{\aword}$ encodes a computation path of~$\ATuring$ if and only if  $\widetilde{\aword} = \aword_1 \cdot {\dots} \cdot \aword_{d}$, for some words $\aword_1,\dots,\aword_d \in (\Alphabet \cup \States)^*$ such that
\begin{itemize}
    \item for every $i \in [1,d]$, $\aword_i \in C$,
    \item for every $i \in [1,d-1]$, $c(\aword_{i+1}) \in \Delta(c(\aword_i))$.
\end{itemize}
Each word in $\overline{\Alphabet}^*$ encodes at most one computation path of~$\ATuring$. 
Moreover, notice that given two words~$\overline{\aword}_1,\overline{\aword}_2 \in \overline{\Alphabet}^*$,
if $\overline{\aword}_1$ has length at least $h(\card{\aword})^2$,
then $\overline{\aword}_1$ and $\overline{\aword}_1\cdot\overline{\aword}_2$ encode the same computation path (if they encode one).

We remind the reader that given a \nTM~$\Turing$ working on an alphabet~$\Alphabet$, we write 
$\Turing(\aworld_1,\dots,\aworld_n)$ for the run of the \nTM on input $(\aworld_1,\dots,\aworld_n) \in \Alphabet^n$. We extend this notation and, given $m \leq n$, write $\Turing(\aword_1,\dots,\aword_m)$ for~the run of~$\Turing$ on 
input~${(\aworld_1,\dots,\aworld_m,\epsilon,\dots,\epsilon) \in \Alphabet^n}$, 
i.e.~an input where the last $n - m$ tapes are empty.
Given two inputs $(\aworld_1,\dots,\aworld_n) \in \Alphabet^n$ and $(\aworld_1',\dots,\aworld_n') \in \Alphabet^n$
we say that $\Turing(\aworld_1,\dots,\aworld_n)$ and $\Turing(\aworld_1',\dots,\aworld_n')$ \defstyle{perform the same computational steps} if the sequence of states produced during the runs $\Turing(\aworld_1,\dots,\aworld_n)$ and $\Turing(\aworld_1',\dots,\aworld_n')$ is the same. 
In particular, given~$t \in \Nat$, 
$\Turing(\aworld_1,\dots,\aworld_n)$ accepts (resp.~rejects) in time~$t$ if and only if $\Turing(\aworld_1',\dots,\aworld_n')$ accepts (resp.~rejects) in time~$t$.

One last notion is desirable in order to prove~\Cref{theorem:kaexp-alternation-complexity}. 
We say that a computation path $\pi = (c_0,\dots,c_d)$ of $\ATuring$, is an (existential or universal) \emph{hop}, if one of the following holds:
\begin{itemize}[align=left]
    \item[\textit{(existential hop)}:] all the states in the configurations $c_0,\dots,c_{d-1}$ belong to~$\States_{\exists}$, and the state of $c_d$ belongs to $\States_{\forall} \cup \{\qacc,\qrej\}$,
    \item[\textit{(universal hop)}:] all the states in the configurations $c_0,\dots,c_{d-1}$ belong to~$\States_{\forall}$, and the state of $c_d$ belongs to $\States_{\exists} \cup \{\qacc,\qrej\}$,
\end{itemize}
Intuitively, in a hop $\pi = (c_0,\dots,c_d)$, no alternation occurs in the computation path $(c_0,\dots,c_{d-1})$ and either $\pi$ is terminating or an alternation occurs between $c_{d-1}$ and $c_d$.

We are now ready to state the essential technical lemma (\Cref{lemma:almost-theorem-kaexp-alternation}) that allows us to prove~\Cref{theorem:kaexp-alternation-complexity}. 
For simplicity, the lemma is stated referring to~$k$,~$f$,~$g$,~$h$,~$\ATuring$, $\aword$ and $\overline{\Alphabet}$ as defined above.

\begin{lemma}
    \label{lemma:almost-theorem-kaexp-alternation}
    Let $m = g(\card{\aword})$ and $n \geq m + 3$.
    One can construct in polynomial time in $n$, $k$,~$\card{\ATuring}$ and $\card{\aword}$ a~\nTM~$\Turing$ working on alphabet~$\overline{\Alphabet}$ with blank symbol~$\trail$, and such that 
    \begin{enumerate}[label=\rm{\textbf{\Roman*}}]
        \setlength{\itemsep}{3pt}
        \item\label{lem:almost-theorem-prop-1} $\Turing$ runs in polynomial time on the length of its inputs,
        \item\label{lem:almost-theorem-prop-2} given $i \in [1,m]$ and $\aword_1,\dots,\aword_i,\dots,\aword_m,\aword_i' \in \overline{\Alphabet}^*$, such that~$\card{\aword_i} \geq h(\card{\aword})^2$,
        $\Turing(\aword_1,\dots,\aword_m)$ 
        and $\Turing(\aword_1,\dots,\aword_i \cdot \aword_i', \dots \aword_m)$ perform the same computational steps,
        \item\label{lem:almost-theorem-prop-2b} for all $\aword_1,\dots,\aword_m \in \overline{\Alphabet}^*$ and $\aword_{m+1},\aword_{m+1}',\dots,\aword_n,\aword_n' \in \overline{\Alphabet}^*$, $\Turing(\aword_1,\dots,\aword_m,\aword_{m+1},\dots,\aword_n)$ 
        and $\Turing(\aword_1,\dots,\aword_m,\aword_{m+1}',\dots,\aword_n')$ perform the same computational steps,
        \item\label{lem:almost-theorem-prop-3} $\aword \in \alang(\ATuring)$ \ iff \
            $Q_1 \aworld_1 \in \overline{\Alphabet}^{h(\card{\aword})^2},\dots, Q_m \aworld_m \in \overline{\Alphabet}^{h(\card{\aword})^2}$ : 
            $\Turing(\aworld_1,\dots,\aworld_m)$ accepts,
    \end{enumerate}
    where for every $i \in [1,m]$, $Q_i = \exists$ if $i$ is odd, and otherwise $Q_i = \forall$.
\end{lemma}

\begin{proof}
    Roughly speaking, the machine~$\Turing$ we shall define checks if the input~$(\aword_1,\dots,\aword_m)$ encodes an accepting computation path in $\ATuring$, where each 
    $\aword_i$ represents an hop.

    Notice that the lemma imposes~$n \geq m + 3$, with $m = g(\card{\aword})$. This is done purely for technical convenience, as it allows us to rely on three additional tapes, i.e.~the $n$-th, $(n-1)$-th and $(n-2)$-th tapes, whose initial content does not affect the run of~$\Turing$ (see property~\eqref{lem:almost-theorem-prop-2b}).
    These three tapes are used by~$\Turing$ as auxiliary tapes.
    To simplify further the presentation of the proof without loss of generality, we assume that $\Turing$ features \emph{pseudo-oracles} that implement the standard unary functions on natural numbers~$(\,. + 1)$~$\ceil{\log(.)}$,
    $\floor{\sqrt{.}}$, as well as the function~$(\,. > f(\card{\aword}))$.
    As in the case of oracles, the machine~$\Turing$ can call a pseudo-oracle, which read the \defstyle{initial content} of the current tape of~$\Turing$ (i.e.~the word in $(\overline{\Alphabet}\setminus \{\trail\})^*$ that occurs the beginning of the tape currently scanned by the read/write head, and delimited by the first occurrence of~$\trail$) and produce the result of the function they implement at the beginning of the (auxiliary) $n$-th tape.
    However, all the functions implemented by pseudo-oracles 
    can be computed in polynomial time, and one can construct in polynomial time Turing machines that compute them. 
    More precisely, $\BigO{1}$ states are needed to implement~$(\,. + 1)$,~$\ceil{\log(.)}$ and $\floor{\sqrt{.}}$, and $\BigO{f(\card{\aword})}$ states are needed in order to implement $(\,. > f(\card{\aword}))$.%
    \footnote{Given an input written in unary, 
        the only non-standard function~$(\,. > f(\card{\aword}))$ can be implemented with a Turing machine that runs in linear time on the size of its input and relies on a chain of~$f(\card{\aword})+1$ states, plus the accepting state and rejecting state.
        At the $i$-th step of the computation, with~$i \leq f(\card{\aword})$, the read/write head is on the $i$-th symbol~$c$ of the input tape and the machine is in its $i$-th state. 
        If~$c$ is non-blank, the machine moves the read/write head to the right and switch to the $(i+1)$-th state. Otherwise, the machine rejects the input.
        On the last state (reached at step $i = f(\card{\aword})+1$, if the machine did not previously reject), if the head reads a non-blank symbol, then the machine accepts, otherwise it rejects.
    }
    So, the pseudo-oracles \emph{can be effectively removed} by incorporating their equivalent Turing machine directly inside~$\Turing$, only growing its size polynomially, and resulting in~$\Turing$ being a standard~\nTM.
    \Cref{proofT8:pseudo-oracles} formalises the semantics of the pseudo-oracles, given the initial content $c \in (\overline{\Alphabet} \setminus \{\trail\})^*$ of the current tape.
    Without loss of generality, we also assume~$\card{\aword}$ and $f(\card{\aword})$ to be at least~$1$. 
    Lastly, again in order to simplify the presentation and without loss of generality, we assume that~$\Turing$ is able to reposition the read/write head to the first position of a given tape.
    This is a standard assumption, which can be removed by simply adding a new symbol to the alphabet (historically,~\textschwa), which shall precede the inputs of the machine. Then, the machine can retrieve the first (writeable) position on the tape by moving the read/write head to the left until it reads the new symbol~\textschwa, to then move right once, without overwriting~\textschwa.

    \begin{figure}
        \flushright 
    \begin{minipage}{0.87\linewidth}
        \begin{itemize}[align=right,nosep,itemsep=5pt,before=\vspace{2pt},after=\vspace{2pt}]
            \item[$(\,. +1)$ :] write~$c \cdot \asymbol \cdot \trail$ at the beginning of the $n$-th tape, where $\asymbol \in \overline{\Alphabet} \setminus \{\trail\}$ is a fixed symbol. So, the length of the initial content on the $n$-th tape becomes $\card{c}+1$.
            \item[$\ceil{\log(.)}$ :] write $a^{\ceil{\log(\card{c})}}\cdot\trail$ at the beginning of the $n$-th tape, where $\asymbol \in \overline{\Alphabet} \setminus \{\trail\}$ is fixed.
            So, the length of the initial content of the $n$-th tape becomes~$\ceil{\log(\card{c})}$.
            \item[$\floor{\sqrt(.)}$ :] write $a^{\floor{\sqrt(\card{c})}}\cdot\trail$ at the beginning of the $n$-th tape, where $\asymbol \in \overline{\Alphabet} \setminus \{\trail\}$ is fixed.
            So, the length of the initial content of the $n$-th tape becomes~$\floor{\sqrt(\card{c})}$.
            \item[$(\,. > f(\card{\aword}))$ :] if $\card{c} > f(\card{\aword})$, write $\asymbol \cdot \trail$ at the beginning of the $n$-th tape, where $\asymbol \in \overline{\Alphabet} \setminus \{\trail\}$ is fixed. Else, write $\trail$ at the beginning of the $n$-th tape. So, the initial content of the $n$-th tape becomes empy if and only if $\card{c} \leq f(\card{\aword})$.
        \end{itemize}
    \end{minipage}
    \caption{Pseudo-oracles; $c \in (\overline{\Alphabet} \setminus \{\trail\})^*$ is the initial content of the current tape.}
    \label{proofT8:pseudo-oracles}
    \end{figure}

    Before describing~$\Turing$, we notice that property~\eqref{lem:almost-theorem-prop-2}
    requires to check whether~$\card{\aword_i} \geq h(\card{\aword})^2$. We rely on~\Cref{lemma:tetra-property-3,lemma:isqrt} in order to perform this test in polynomial time on $\card{\aword_1}$,
    without computing $h(\card{\aword})$ explicitly. 
    We have, 
    \begin{align*}
        \card{\aword_i} \geq h(\card{\aword})^2
        & \text{ iff } \ \isqrt{\card{\aword_i}} \geq h(\card{\aword}),
        & \text{ by~\Cref{lemma:isqrt}}\\[3pt]
        & \text{ iff } \ \isqrt{\card{\aword_i}} + 1 > \tetra(k,f(\card{\aword})),\\[3pt]
        & \text{ iff } \ \klog{k}(\isqrt{\card{\aword_i}} + 1) > f(\card{\aword}).
        & \text{ by~\Cref{lemma:tetra-property-3}}
    \end{align*}
    To check whether $\klog{k}(\isqrt{\card{\aword_i}} + 1) > f(\card{\aword})$ holds, 
    it is sufficient to write $\aword_i$
    on the $n$-th tape, and then call first the pseudo-oracle for $\floor{\sqrt{.}}$, followed by the pseudo-oracle for $(\,. + 1)$ and by $k$ calls to the pseudo-oracle for~$\ceil{\log{(.)}}$ (recall that $k$ is written in unary), and lastly one call to the pseudo-oracle for $(\, . > f(\card{\aword}))$.
    When taking into account the number of states needed in order to implement these pseudo-oracles, we conclude that the function $\klog{k}(\isqrt{.} + 1) > f(\card{\aword})$ can be implemented by a Turing machine with~$\BigO{k + f(\card{\aword})}$ states.
    Hence, below we assume without loss of generality that~$\Turing$ also features pseudo-oracles for the function~$\klog{k}(\isqrt{.} + 1) > f(\card{\aword})$
    as well as the (similar) functions~$\klog{k}(\,. + 1) > f(\card{\aword})$ 
    and~$\klog{k}(\,. \,) > f(\card{\aword})$,
    both requiring~$\BigO{k + f(\card{\aword})}$ states
    to be implemented. Ultimately, because of the polynomial bounds on the number of states required to implement these functions,~$\Turing$ (without pseudo-oracles) can be constructed in \ptime on~$n$,~$k$,~$\card{\aword}$ and~$\card{\ATuring}$.

    We are now ready to construct the~\nTM~$\Turing$ that works on alphabet~$\overline{\Alphabet}$. For an input ${(\aword_1,\dots,\aword_n) \in (\overline{\Alphabet}^*})^n$, 
    $\Turing$ operates following $m$ ``macrosteps'' on the first $m$ tapes, disregarding the input of the last $n-m \geq 1$ tapes. 
    At step $i \in [1,m]$, the machine operates as follows:
    \begin{description}
        \setlength{\itemsep}{3pt}
        \item[macrostep $i = 1$.] Recall that~$Q_1 = \exists$. $\Turing$ works on the first tape.
            \begin{enumerate}
                \setlength{\itemsep}{3pt}
                \item\label{turing:step1:substep1} 
                $\Turing$ reads the prefix~$\widetilde{\aword}$ of~$\aword_1$, such that if $\card{\aword_1} \leq h(\card{\aword})^2$ then $\widetilde{\aword} = \aword_1$, else $\card{\widetilde{\aword}} = h(\card{\aword})^2$. (\textit{Note: $\Turing$ does not consider the part of the input after the $h(\card{\aword})^2$-th symbol})
                \item\label{turing:step1:substep2} The \nTM~$\Turing$ then checks whether $\widetilde{\aword}$ corresponds to an \emph{existential hop} in~$\ATuring$ starting from the state~$\qinit$. 
                It does so by analysing from left to right~$\widetilde{\aword}$ in chunks of length $h(\card{\aword})$ (i.e.~the size of a word in the language~$C$ of the configurations of~$\ATuring$), possibly with the exception of the last chunk, which can be of smaller size. 
                Let $d$ be the number of chunks, so that $\widetilde{\aword} = \widetilde{\aword}_1 \cdot {\dots} \cdot \widetilde{\aword}_d$, where $\widetilde{\aword}_j$ is the $j$-th chunk analysed by~$\Turing$.
                    \begin{enumerate}[label=(\alph*)]
                        \setlength{\itemsep}{3pt}
                        \item If $\widetilde{\aword}_j \not\in C$, $\Turing$ \textbf{rejects}, (\textit{Note: else, $\widetilde{\aword}_j$ encodes the configuration $c(\widetilde{\aword}_j)$~of~$\ATuring$})
                        \item if the symbol~$\qinit$ does not occur in the first chunk $\widetilde{\aword}_1$, then $\Turing$ \textbf{rejects}, 
                        \item if $\widetilde{\aword}_j$ is not the last chunk and contains a symbol from~$\States_\forall$, then $\Turing$ \textbf{rejects},
                        \item if $\widetilde{\aword}_j$ is not the last chunk and $c(\widetilde{\aword}_{j+1}) \not \in \Delta(c(\widetilde{\aword}_j))$, then $\Turing$ \textbf{rejects}.
                        \item If the length of the last chunk~$\widetilde{\aword}_d$ is not $h(\card{\aword})$ (i.e.~the length of~$\widetilde{\aword}$ is not a multiple of~$h(\card{\aword})$), then $\Turing$ \textbf{rejects}, (\textit{Note: this case subsumes the case where~$\aword_1$ is empty})
                        \item if the last chunk $\widetilde{\aword}_d$ contains a symbol from $\States_{\exists}$, then $\Turing$ \textbf{rejects}.
                    \end{enumerate}
                \item\label{turing:step1:substep3} $\Turing$ analyses the last chunk $\widetilde{\aword}_d$ of the previous step. 
                If~$\qacc$ occurs in $\widetilde{\aword}_d$, then $\Turing$ \textbf{accepts}. 
                Otherwise, if either~$m = i = 1$ or $\widetilde{\aword}_d$ contains the symbol~$\qrej$, 
                then $\Turing$ \textbf{rejects}.
                Else,~$\Turing$ writes down $\widetilde{\aword}_d \cdot \trail$ at the beginning of the first tape, and \textbf{moves to macrostep~$2$}. Let $\aword^{(1)} = \widetilde{\aword}_d$, i.e.~the word (currently store at the beginning of the first tape) encoding the last configuration of the computation path encoded by~$\widetilde{\aword}$.
            \end{enumerate}

            From Step~\eqref{turing:step1:substep1}, we establish the following property of macrostep~1.

            \begin{claim}
                \label{claim:turing:macrostep-1-input}
                Macrostep~1 only reads at most the first $h(\card{\aword})^2$ character of the input of the first tape, and does not depend on the input written on other tapes.  
            \end{claim}

            We analyse in details the complexity of macrostep~1.

            \begin{claim}
                \label{claim:turing:macrostep-1}
                The macrostep~$1$ runs in polynomial time on~$\card{\aword_1}$,~$k$, $\card{\aword}$ and $\card{\ATuring}$,
                and only required polynomially many states to be implemented, with respect to~$k$, $\card{\aword}$ and $\card{\ATuring}$. 
            \end{claim}

            \begin{claimproof}
            We give a more fine grained description of the various steps, and analyse their complexity. 
            Recall that the tapes $n$, $n-1$ and $n-2$ are auxiliary.
            \defstyle{Step~\eqref{turing:step1:substep1}} can be performed in polynomial time in~$\card{\aword_1}$, $k$, $\card{\aword}$ and $\card{\overline{\Alphabet}} \leq \card{\ATuring}$, as shown below.
\begin{MyCode}
write ${\trail}^{\card{\aword_1}+1}$ at the beginning the $(n-1)$-th tape
while true 
    //Invariant:    The $(n-1)$-th tape contains a prefix of $\aword_1$ 
    //              of length less than $\min(\card{\aword_1},h(\card{\aword})^2)$
    let $i$ be the length of the initial content of the $(n-1)$-th tape
    read the $i$-th symbol $\asymbol$ of the first tape
    if $\asymbol = \trail$, break
    write $\asymbol$ in the $i$-th position of the $(n-1)$-th tape
    call the pseudo-oracle for $\smash{\klog{k}(\isqrt{.} + 1) > f(\card{\aword})}$ on the $(n-1)$-th tape 
    if the initial content of the $n$-th tape is not empty, break 
\end{MyCode}
            At the end of this procedure, the $(n-1)$-th tape contains the word~$\widetilde{\aword}$.
            Indeed, lines 3, 4 and 6 copy the $i$-th character of $\aword_1$ to the $(n-1)$-th tape. 
            If $\card{\aword_1} < h(\card{\aword})^2$, then
            $\widetilde{\aword} = \aword_1$ and
            after copying $\aword_1$ on the $(n-1)$-th tape, the test in line 5 becomes true. 
            Otherwise, after copying the first $h(\card{\aword})^2$ characters of $\aword_1$ to the $(n-1)$-th tape, the pseudo-oracle call of line 7 will write a non-blank symbol on the $n$-th tape, making the test in line 8 true.
            Time-wise, the complexity of the procedure above is polynomial in~$\card{\aword_1}$, $k$ (line 9), $\card{\aword}$ (line 9) and $\card{\overline{\Alphabet}} \leq \card{\ATuring}$ (lines 5--8). Space-wise, performing line 1 of the procedure requires a constant number of states. 
            Performing lines 5--8 requires $\BigO{\card{\overline{\Alphabet}}}$ states, as~$\Turing$ needs to keep track of which symbol was read on the first tape, in order to write in on the tape $(n-1)$.
            For line 9 instead, we must take into account the number of states required to implement the computation done by the pseudo-oracle directly in~$\Turing$. As stated at the beginning of the proof, these are $\BigO{k + f(\card{\aword})}$ many states. Hence, $\Turing$ can implement the procedure above with $\BigO{k + f(\card{\aword}) + \card{\ATuring}}$ many states.

            Let us now move to 
            \defstyle{Step~\eqref{turing:step1:substep2}}.
            Recall that, after 
            \defstyle{Step~\eqref{turing:step1:substep1}},~$\widetilde{\aword}$ is stored at the beginning of tape $(n-1)$.
            The machine~$\Turing$ continues as follows, where the lines marked with symbols of the form \ $\pgftextcircled{i}$ \ are later expanded and analysed in details.

\begin{MyCode}
write $\widetilde{\aword} \cdot \trail$ at the beginning of the first tape
if the initial content of the first tape is empty, $\text{\rm \textbf{reject}}$
$\pgftextcircled{1}$ if $\qinit$ does not appear in the first $h(\card{\aword})$ characters of $\widetilde{\aword}$, $\text{\rm \textbf{reject}}$
while true 
    //Invariant:   The 1st tape starts with a non-empty suffix of $\widetilde{\aword}$.
    let $i$ be the length of the initial content of the first tape
    write $\trail^{i+1}$ on both the $(n-1)$-th and $(n-2)$-th tapes
    call the pseudo-oracle for $\smash{\klog{k}(\, . + 1) > f(\card{\aword})}$ on the first tape 
    if the initial content of the $n$-th tape is empty, $\text{\rm{\textbf{rejects}}}$
    $\pgftextcircled{2}$ copy the first $h(\card{\aword})$ characters of the first tape to the $(n-1)$-th tape
    //The tape $(n-1)$ now contains a chunk $\widetilde{\aword}_j$.
    if the initial content of the $(n-1)$-th tape and the first tape is equal, break
    $\pgftextcircled{3}$ if the initial content of the $(n-1)$-th tape does not belong to $C$, $\text{\rm \textbf{reject}}$
    $\pgftextcircled{4}$ if a symbol from $\States_\forall$ occur in the initial content of the $(n-1)$-th tape, $\text{\rm \textbf{reject}}$
    $\pgftextcircled{5}$ shift the initial content of the first tape $h(\card{\aword})$ positions to the left, effectively removing the first $h(\card{\aword})$ characters from the tape 
    //The first tape now starts with the chunk $\widetilde{\aword}_{j+1}$.
    $\pgftextcircled{6}$ copy the first $h(\card{\aword})$ characters of the first tape to the $(n-2)$-th tape
    //The tape $(n-2)$ now contains the chunk $\widetilde{\aword}_{j+1}$.
    $\pgftextcircled{7}$ perform step (2d) //$\widetilde{\aword}_j$ and  $\widetilde{\aword}_{j+1}$ on tapes $(n-1)$ and $(n-2)$, respectively
//Post: The tape $(n-1)$ contains the last chunk $\widetilde{\aword}_d$.
$\pgftextcircled{8}$ if the initial content of the $(n-1)$-th tape does not belong to $C$, $\text{\rm \textbf{reject}}$
$\pgftextcircled{9}$ if a symbol from $\States_\exists$ occur in the initial content of the $(n-1)$-th tape, $\text{\rm \textbf{reject}}$
\end{MyCode}

Provided that the lines marked with symbols of the form \ $\pgftextcircled{i}$ \ can be performed in polynomial time and can be implemented with a polynomial number of states, it is straightforward to see that this procedure also runs in polynomial time and only uses polynomially many states.
Indeed, at the $j$-th iteration of the while loop, the $j$-th chunk of $\widetilde{\aword}$ is analysed. 
Therefore, the body of the while loop is executed at most $d = \ceil{\frac{\card{\widetilde{\aword}}}{h(\card{\aword})}}$ times, where $d$ is the number of chunks (which is polynomial in~$\card{\aword_1}$).
Moreover, it is easy to see that the lines 1,2, 6--9 and 12 runs in polynomial time on $\card{\aword_1}$, $k$, $\card{\aword}$ and $\card{\ATuring}$, and only requires $\BigO{k+f(\card{\aword}) + \card{\ATuring}}$ many states to be implemented.
Indeed, line 1 can be implemented in linear time on $\card{\widetilde{\aword}} \leq \card{\aword_1}$ and $\card{\overline{\Alphabet}}$, and requires $\BigO{\overline{\Alphabet}}$ states to be implemented, in order to track the symbols read on the $(n-1)$-th tape, and copy them on the first tape.
A similar analysis can be done for line 12.
Line~2 runs in time $\BigO{1}$ and requires $\BigO{1}$ states.
Line~7 runs in time $\BigO{\card{\widetilde{\aword}}}$ and requires $\BigO{1}$ states.
Line~8 is analogous to line~9 of Step~\eqref{turing:step1:substep1}, and checks whether the current suffix of~$\widetilde{\aword}$ stored in the first tape contains at least $h(\card{\aword})$ characters. 
Together with line \ $\pgftextcircled{5}$ \,, 
this line assures that $\card{\widetilde{\aword}}$ is a multiple of~$h(\card{\aword})$, effectively implementing the step (2e).

Let us now focus on the lines marked with \ $\pgftextcircled{i}$ \,. Line \ $\pgftextcircled{1}$ \,, implements the step~(2b). To test whether $h(\card{\aword})$ character have been read, we can rely on 
the pseudo-oracle for $\smash{\klog{k}(\, . + 1) > f(\card{\aword})}$, similary to what it is done in line 7. Here is the resulting procedure.
\begin{MyCode}
write $\trail^{\card{\widetilde{\aword}}+1}$ on the $(n-1)$-th tape
while true 
    //Invariant:    The $(n-1)$-th tape contains a prefix of $\widetilde{\aword}$ 
    //              of length less than $\min(\card{\widetilde{\aword}},h(\card{\aword}))$
    let $i$ be the length of the initial content of the $(n-1)$-th tape
    read the $i$-th symbol $\asymbol$ of the first tape
    if $\asymbol = \qacc$, break
    if $\asymbol = \trail$, $\text{\rm\textbf{reject}}$
    write $\asymbol$ in the $i$-th position of the $(n-1)$-th tape
    call the pseudo-oracle for $\smash{\klog{k}(\,. + 1) > f(\card{\aword})}$ on the $(n-1)$-th tape 
    if the initial content of the $n$-th tape is not empty, $\text{\rm\textbf{reject}}$ 
//Post: $\qacc$ appears in the first $h(\card{\aword})$ characters of $\widetilde{\aword}$
\end{MyCode}
With a similar analysis to what done in Step~\eqref{turing:step1:substep1}, we conclude that \ $\pgftextcircled{1}$ \, runs in polynomial time on $\card{\aword_1}$, $k$, $\card{\aword}$ and $\card{\ATuring}$, and requires $\BigO{k+f(\card{\aword}) + \card{\ATuring}}$ states to be implemented.

Line~~$\pgftextcircled{2}$ \ is very similar to line \ $\pgftextcircled{1}$\,. Notice that form line 6 of Step~\eqref{turing:step1:substep2}, the initial content of the tape $(n-1)$ is empty. Moreover, from line 7, the initial content of the first tape has at least $h(\card{\aword})$ characters.
Here is the procedure for~~$\pgftextcircled{2}$\,. 
\begin{MyCode}
while true 
    //Invariant:    The $(n-1)$-th tape contains a prefix of the initial content 
    //              of the first tape, of length less than $h(\card{\aword})$
    let $i$ be the length of the initial content of the $(n-1)$-th tape
    read the $i$-th symbol $\asymbol$ of the first tape
    write $\asymbol$ in the $i$-th position of the $(n-1)$-th tape
    call the pseudo-oracle for $\smash{\klog{k}(\,. + 1) > f(\card{\aword})}$ on the $(n-1)$-th tape 
    if the initial content of the $n$-th tape is not empty, break
//Post: the initial content of the $(n-1)$-th tape is a prefix of the 
//       initial content of the first tape, of length $h(\card{\aword})$
\end{MyCode}
As in the case of line~~$\pgftextcircled{1}$\,,
this procedure runs in polynomial time on $\card{\aword_1}$, $k$ $\card{\aword}$ and $\card{\ATuring}$, 
and requires $\BigO{k+f(\card{\aword}) + \card{\ATuring}}$ states to be implemented.
After executing line~~$\pgftextcircled{1}$\,, the initial content of the tape $(n-1)$ is a chunk of $\widetilde{\aword}$, say $\widetilde{\aword}_j$. 

Lines \ $\pgftextcircled{3}$ \ and \ $\pgftextcircled{4}$ \ 
implement the steps (2a) and (2c).
The machine~$\Turing$ can implement both lines simultaneously,
by iterating through the initial content of the $(n-1)$-th tape as shown below.
\begin{MyCode}
move the read/write head to the first position of the $(n-1)$-th tape
while true 
    let $\asymbol$ be the symbol corresponding to the read/write head
    if $\asymbol = \trail$, $\text{\rm{\textbf{reject}}}$ 
    if $\asymbol \in \States_{\forall}$, $\text{\rm{\textbf{reject}}}$
    if $\asymbol \in \States \setminus \States_{\forall}$, break
    move the read/write head to the right
//Post: a symbol from $\States \setminus \States_{\forall}$ was found
move the read/write head to the right 
if the the read/write head reads $\trail$, $\text{\rm{\textbf{reject}}}$
while the read/write head does not read $\trail$
    if the read/write had reads a symbol from $\States$, $\text{\rm{\textbf{reject}}}$ 
    move the read/write head to the right
//Post: the initial content of the $(n-1)$-th tape belongs to $C$ and does not 
//       contain symbols from $\States_{\forall}$
\end{MyCode}
The correctness of this procedure with respect to the description of~lines~~$\pgftextcircled{3}$ \, and~~$\pgftextcircled{4}$ \, is straightforward as we recall that $C = \{ \widehat{\aword} \in \Alphabet^* \cdot \States \cdot \Alphabet^+ \mid \card{\widehat{\aword}} =  h(\card{\aword}) \}$, and by line~~$\pgftextcircled{1}$ \, the initial content of the tape $(n-1)$ has length $h(\card{\aword})$.
The running time of the procedure above is polynomial in the length of the initial content of the $(n-1)$-th tape, and thus polynomial on $\card{\aword_1}$. 
The number of states required to implement this procedure is in~$\BigO{1}$.

Let us move to line~~$\pgftextcircled{5}$\,. Essentially, in this lines~$\Turing$ removes from the initial content of the first tape the chunk that is currently analysed and stored in the tape $(n-1)$, so that the initial content of the first tape starts with the next chunk. A possible implementation of this line is given below. Recall that, from line 6 of Step~\eqref{turing:step1:substep2}, the $(n-2)$-th tape starts with the word $\trail^{i+1}$, where $i$ is the length of the initial content of the first tape. 
Moreover, from line 12 of Step~\eqref{turing:step1:substep2}, 
the length of the initial content of the first tape is greater than the length of the initial content of the tape $(n-1)$.
\begin{MyCode}
move the read/write head to the first occurrence of $\trail$ on the $(n-1)$-th tape 
//The read/write head is now on the position $\card{\widetilde{\aword}_j}$ of the $(n-1)$-th tape
move the read/write head to the first tape 
//Read/write head currently in position $\card{\widetilde{\aword}_j}$ of the first tape
//From line 12 of Step (2), the head reads a symbol different from $\trail$
while true
    let $\asymbol$ be the symbol corresponding to the read/write head 
    write $\trail$
    move the read/write head to the right 
    if the read/write head reads $\trail$,
        move the read/write head to the first occurrence of $\trail$ on the $(n-2)$-th tape
        write $\asymbol$
        break
    else 
        move the read/write head to the first occurrence of $\trail$ on the $(n-2)$-th tape
        write $\asymbol$
        move the read/write head to the first occurrence of $\trail$ on the first tape
        //read/write head in the position where~$\asymbol$ was previously stored
        move the read/write head right
//Post: The initial content of the tape $(n-2)$ is the word required by line $\pgftextcircled{5}$
let $\widetilde{\aword}'$ be the initial content of the $(n-1)$-th tape
write $\widetilde{\aword}' \cdot \trail$ at the beginning of the first tape
\end{MyCode}
Again, this procedure runs in polynomial time on~$\card{\widetilde{\aword}}$ and $\card{\ATuring}$. Since~$\Turing$ needs to keep track of the symbol read in line 7, as well as implementing lines~21--22, the number of states required to implement the procedure is $\BigO{\overline{\Alphabet}}$, and thus polynomial in $\card{\ATuring}$.

Line~~$\pgftextcircled{6}$ \ copies the chunk $\widetilde{\aword}_{j+1}$ to the tape $(n-2)$, so that $\Turing$ can then perform the step (2d) (line~~$\pgftextcircled{7}$\,). Line~~$\pgftextcircled{6}$ \ is implemented analogously to line~~$\pgftextcircled{2}$\,, i.e. the line that copied the chunk $\widetilde{\aword}_j$ to the $(n-1)$-th tape.
Compared to the code for line~~$\pgftextcircled{2}$\,, it is sufficient to replace the steps done to the tape $(n-1)$ to equivalent steps done on the $(n-2)$ step in order to obtain the code for line~~$\pgftextcircled{6}$\,, which therefore runs in polynomial time, and requires $\BigO{k + f(\card{\aword}) + \card{\ATuring}}$ many sstates in order to be implemented.

Let us now assume that the initial content of the $(n-1)$-th step is the chunk $\widetilde{\aword}_j$, and that the initial content of tape $(n-2)$ is the chunk $\widetilde{\aword}_{j+1}$.
According to line~~$\pgftextcircled{7}$\,, we must check whether $c(\widetilde{\aword}_{j+1}) \in \Delta(c(\widetilde{\aword}_j))$ (step (2d)). Notice that we do now know whether $\widetilde{\aword}_{j+1} \in C$, as this test will be performed at the next iteration of the while loop of Step~\eqref{turing:step1:substep2}. 
However, we can still easily implement step (2d) as follows: 
we find the pair $(\astate,\asymbol) \in \States \times \Alphabet$ such that $\aword' \cdot \astate \cdot \asymbol \cdot \aword'' = \widetilde{\aword}_j$ for some words $\aword'$ and $\aword''$. This decomposition is guaranteed from $\widetilde{\aword}_j \in C$, which holds from line~~$\pgftextcircled{3}$\,.
The machine then iterates over all elements of $\delta(\astate,\asymbol)$, updating $\widetilde{\aword}_j$ accordingly, on the $n$-th tape. After each update, $\Turing$~checks whether the initial content of the $n$-th tape corresponds to $\widetilde{\aword}_{j+1}$. Here is the procedure.

\begin{MyCode}
let $\widetilde{\aword}_j$ be the initial content of the $(n-1)$-th tape
let $\widetilde{\aword}_{j+1}$ be the initial content of the $(n-2)$-th tape
move the read/write head to the beginning of the tape $(n-1)$
while the read/write had reads a character not from $\States$
    move the read/write head to the right 
//Since $\widetilde{\aword}_j \in C$, the loop above terminates
//Post: the read/write head reads a character from $\States$
let $\astate$ be the symbol corresponding to the read/write head
move the read/write head to the right 
let $\asymbol$ be the symbol corresponding to the read/write head
//Since $\widetilde{\aword}_j \in C$, $\asymbol$ belongs to $\Alphabet$
for $(\astate',\asymbol',i) \in \States \times \Alphabet \times \{-1,+1\}$ belonging to $\delta(\astate,\asymbol)$
    write $\widetilde{\aword}_j \cdot \trail$ at the beginning of the $n$-th tape
    move to the read/write tape to the (only) occurrence of $s$ on the $n$-th tape
    if $i = 1$,
        write $\asymbol'$ //overwrites $\astate$
        move the read/write tape to the right
        write $\asymbol'$ //overwrites $\asymbol$
    else if $s$ occurs at the beginning of the $n$-th tape,
        write $\qrej$
    else
        move the read/write tape left
        let $\asymbolbis$ be the symbol corresponding to the read/write head
        write $\astate'$ //overwrites $\asymbolbis$
        move the read/write tape to the right
        write $\asymbolbis$ //overwrites $\astate$
        move the read/write tape to the right 
        write $\asymbol'$ // overwrites $\asymbol$
    if the initial content of the $n$-th tape equals $\widetilde{\aword}_{j+1}$, goto line 32
//Post : $c(\widetilde{\aword}_{j+1}) \not\in \Delta(c(\widetilde{\aword}_j))$ 
$\text{\rm{\textbf{reject}}}$
//If this line is reached, $c(\widetilde{\aword}_{j+1}) \in \Delta(c(\widetilde{\aword}_j))$
\end{MyCode}
Briefly, if the test in line 14 holds, then $\widetilde{\aword}_j = \aword' \cdot \astate \cdot \asymbol \cdot \aword''$ is copied on the $n$-th tape and updated to $\aword' \cdot \asymbol' \cdot \astate' \cdot \aword''$, according to the semantics of the ATM~$\ATuring$ for the case $(\astate',\asymbol',1) \in \delta(\astate,\asymbol)$.
If instead the test in line 19 holds, then $\widetilde{\aword}_j$
is of the form $\astate \cdot \asymbol \cdot \aword''$ we are considering a triple $(\astate',\asymbol',-1) \in \delta(\astate,\asymbol)$. According to the semantics of the $\ATM$ $\widetilde{\aword}_j$ must be updated to $\qrej \cdot a \cdot \aword''$, as done in line 20.
Lastly, if the else branch in line 21 is reached, then $\widetilde{\aword}_j$ is of the form $\aword' \cdot \asymbolbis \cdot \astate \cdot \asymbol \cdot \aword''$, and since we are considering a triple $(\astate',\asymbol',-1) \in \delta(\astate,\asymbol)$, it must be updated to $\aword' \cdot \astate' \cdot \asymbolbis \cdot \asymbol' \cdot \aword''$, as done in lines 22--28.
This procedure runs in polynomial time on~$\card{\widetilde{\aword}_j} \leq \card{\aword_1}$ and $\card{\ATuring}$, and because of the \defstyle{for} loop in line 12 together with the necessity to copy $\widetilde{\aword}_j$ to the $n$-th tape and keeping track of the symbols $\astate$, $\asymbol$ and $\asymbolbis$, it requires a number of states that is polynomial in $\card{\ATuring}$.

Lastly, the lines~~$\pgftextcircled{8}$ \ and~~$\pgftextcircled{9}$ \ are analogous to the lines~~$\pgftextcircled{3}$ \ and~~$\pgftextcircled{4}$\, (the only difference being that~$\States_{\forall}$ is replaced by~$\States_{\exists}$ in order to correctly implement~~$\pgftextcircled{9}$ \ and so (2f)). 
We conclude that Step~\eqref{turing:step1:substep2} of the procedure runs in polynomial time on $\card{\aword_1}$, $k$, $\card{\aword}$ and $\card{\ATuring}$.
and requires a polynomial number of states with respect to $k$, $\card{\aword}$ and $\card{\ATuring}$.  

At the end of Step~\eqref{turing:step1:substep2},
the initial content of the tape $(n-1)$ corresponds to the last chunk~$\widetilde{\aword}_d$ of $\widetilde{\aword}$.  
Then, the step~\eqref{turing:step1:substep3} performed by~$\Turing$ is implemented as follows.
\begin{MyCode}
let $\widetilde{\aword}_d$ be the initial content of the $(n-1)$-th tape
if $\qacc$ appears in $\widetilde{\aword}_d$, $\text{\rm\textbf{accept}}$
if $\qrej$ appears in $\widetilde{\aword}_d$ , $\text{\rm\textbf{reject}}$
if $m = 1$, $\text{\rm \textbf{reject}}$
write $\widetilde{\aword}_d \cdot \trail$ at the beginning of the first tape //see $\aword^{(1)}$
move to macrostep $2$
\end{MyCode}
        Clearly, this step runs in polynomial time on $\card{\widetilde{\aword}_d} \leq \card{\aword_1}$ and $\card{\overline{\Alphabet}} \leq \card{\ATuring}$, and requires $\BigO{\card{\overline{\Alphabet}}}$ many states to be implemented (see line~5).
        Considering all the steps, we conclude that macrostep~$1$ runs in polynomial time on~$\card{\aword_1}$,~$k$, $\card{\aword}$ and $\card{\ATuring}$,
        and requires polynomially many states to be implemented, w.r.t.~$k$, $\card{\aword}$ and $\card{\ATuring}$. 
        \end{claimproof}

            The semantics of~$\Turing$ during macrostep~$1$ is summarised in the following claim.

            \begin{claim}
                \label{claim:turing:macrostep-1-semantics}
                Let $\widetilde{\aword}$ prefix~$\widetilde{\aword}$ of~$\aword_1 \cdot \trail$, 
                s.t.~either $\widetilde{\aword} \cdot \trail$ is a prefix of $\aword_1 \cdot \trail$ or $\card{\widetilde{\aword}} = h(\card{\aword})^2$.
                \begin{itemize}
                    \item if $\widetilde{\aword}$ encodes an existential hop of $\ATuring$ starting on state $\qinit$ and ending in~$\qacc$,
                    $\Turing$ accepts,
                    \item else, if $1 < m$ and $\widetilde{\aword}$ encodes an existential hop of $\ATuring$, starting on state $\qinit$ and ending in a state from~$\States_{\forall}$,
                    then $\Turing$ writes the last configuration of this path at the beginning of the first tape (i.e.~the word~$\aword^{(1)}$) and moves to macrostep 2,
                    \item otherwise, $\Turing$ rejects.
                \end{itemize}
            \end{claim}

        \item[macrostep: $i > 1$, $i$ odd.] Recall that~$Q_i = \exists$. $\Turing$ works on the $i$-th tape and on the word $\aword^{(i-1)}$ written in on the $(i-1)$-th tape at the end of step $i-1$. This macrostep resembles the first one.
        \begin{enumerate}
            \setlength{\itemsep}{3pt}
            \item\label{turing:step2:substep1} 
            $\Turing$ reads the prefix~$\widetilde{\aword}$ of~$\aword_i$, such that if $\card{\aword_i} \leq h(\card{\aword})^2$ then $\widetilde{\aword} = \aword_i$, else $\card{\widetilde{\aword}} = h(\card{\aword})^2$. (\textit{Note: $\Turing$ does not consider the part of the input after the $h(\card{\aword})^2$-th symbol})
            \item\label{turing:step2:substep2} The \nTM~$\Turing$ then checks whether $\aword^{(i-1)} \cdot \widetilde{\aword}$ encodes an existential hop in~$\ATuring$.
            It does so by analysing (from left to right)~$\widetilde{\aword}$ in chunks of length $h(\card{\aword})$, possibly with the exception of the last chunk, which can be of smaller size. 
            Let $d$ be the number of chunks, and let $\widetilde{\aword}_j$ be the $j$-th chunk analysed by $\Turing$.
                \begin{enumerate}[label=(\alph*)]
                    \setlength{\itemsep}{3pt}
                    \item If $\widetilde{\aword}_j \not\in C$, $\Turing$ \textbf{rejects},
                    \item if the first chunk $\widetilde{\aword}_1$ is such that $c(\widetilde{\aword}_1) \not\in \Delta(c(\aword^{(i-1)}))$, then $\Turing$ \textbf{rejects}, 
                    \item if $\widetilde{\aword}_j$ is not the last chunk and contains a symbol from~$\States_\forall$, then $\Turing$ \textbf{rejects},
                    \item if $\widetilde{\aword}_j$ is not the last chunk and $c(\widetilde{\aword}_{j+1}) \not \in \Delta(c(\widetilde{\aword}_j))$, then $\Turing$ \textbf{rejects},
                    \item if the length of the last chunk~$\widetilde{\aword}_d$ is not $h(\card{\aword})$ (i.e.~the length of~$\widetilde{\aword}$ is not a multiple of~$h(\card{\aword})$), then $\Turing$ \textbf{rejects},
                    \item if the last chunk $\widetilde{\aword}_d$ contains a symbol from~$\States_\exists$, then $\Turing$ \textbf{rejects}.
                \end{enumerate}
            \item\label{turing:step2:substep3} $\Turing$ analyses the last chunk $\widetilde{\aword}_d$ of the previous step. 
            If~$\qacc$ occurs in $\widetilde{\aword}_d$, then $\Turing$ \textbf{accepts}. 
            Otherwise, if $i = m$ or $\qrej$ occurs in~$\widetilde{\aword}_d$, then $\Turing$ \textbf{rejects}.
            Else,~$\Turing$ writes down $\widetilde{\aword}_d \cdot \trail$ at the beginning of the $i$-th tape, and \textbf{moves to macrostep~$(i+1)$}. Let $\aword^{(i)} = \widetilde{\aword}_d$.
        \end{enumerate}

        From macrostep~$(i-1)$, 
        the word~$\aword^{(i-1)}$ corresponds to the initial content of the $(i-1)$-th tape, and its length is bounded by $\card{\aword_{i-1}}$.
        Moreover,
        $\aword^{(i-1)} \in C$.
        The step (2b) can be implemented as line~~$\pgftextcircled{7}$ \, in the proof of~\Cref{claim:turing:macrostep-1}.
        The analysis of this macrostep follows essentially the same arguments given for the macrostep~$1$, allowing us to establish the following claims.

        \begin{claim}
            \label{claim:turing:macrostep-i-odd-input}
            Let $i > 1$ odd.
            Macrostep~$i$ only reads at most the first $h(\card{\aword})^2$ characters of the input of the $i$-th tape, and at most the first $h(\card{\aword})$ characters of the $(i-1)$-th tape.
            It does not depend on the content written on the other tapes.
        \end{claim}

        \begin{claim}
            \label{claim:turing:macrostep-i-odd}
            Let $i > 1$ odd. The macrostep~$i$ runs in polynomial time on~$\card{\aword_1}$,~$k$, $\card{\aword}$ and $\card{\ATuring}$,
            and required polynomially many states to be implemented, w.r.t.~$k$, $\card{\aword}$ and $\card{\ATuring}$. 
        \end{claim}

        \begin{claim}
            \label{claim:turing:macrostep-i-odd-semantics}
            Let $i > 1$ odd, and let $\widetilde{\aword}$ prefix~of~$\aword_i \cdot \trail$, 
            such that~either $\widetilde{\aword} \cdot \trail$ is a prefix of $\aword_i \cdot \trail$ or $\card{\widetilde{\aword}} = h(\card{\aword})^2$.
            Suppose that~$\Turing$ did not halt on the macrosteps~$1,\dots,i-1$.
            \begin{itemize}
                \item if $\aword^{(i-1)} \cdot \widetilde{\aword}$ encodes an existential hop of $\ATuring$ ending in state $\qacc$,
                then $\Turing$ accepts,
                \item else, if $i < m$ and $\aword^{(i-1)} \cdot \widetilde{\aword}$ encodes an existential hop of $\ATuring$, ending in a state from~$\States_\forall$,
                then $\Turing$ writes the last configuration of this path (i.e.~$\aword^{(i)}$) at the beginning of the $i$-th tape and moves to macrostep $(i+1)$,
                \item otherwise, $\Turing$ rejects.
            \end{itemize}
        \end{claim}

        \item[macrostep: $i > 1$, $i$ even.] Recall that~$Q_i = \forall$. $\Turing$ works on the $i$-th tape and on the word $\aword^{(i-1)}$ written in on the $(i-1)$-th tape at the end of step $i-1$.
        \begin{enumerate}
            \setlength{\itemsep}{3pt}
            \item\label{turing:step3:substep1}
            $\Turing$ reads the prefix~$\widetilde{\aword}$ of~$\aword_i$, such that if $\card{\aword_i} \leq h(\card{\aword})^2$ then $\widetilde{\aword} = \aword_i$, else $\card{\widetilde{\aword}} = h(\card{\aword})^2$. (\textit{Note: $\Turing$ does not consider the part of the input after the $h(\card{\aword})^2$-th symbol})
            \item\label{turing:step3:substep2} The \nTM~$\Turing$ works ``dually'' with respect to the macrosteps where $i$ is odd. 
            It checks whether $\aword^{(i-1)} \cdot \widetilde{\aword}$ corresponds to an universal hop of~$\ATuring$.
            However, differently from the case where $i$ is odd, if $\widetilde{\aword}$ does not encode such a computation path, then $\Turing$ accepts (as in this case the input is not ``well-formed'' and should not be considered for the satisfaction of the quantifier~$Q_i = \forall$).
            As in the other macrosteps, $\Turing$ perform this step by analysing (from left to right)~$\widetilde{\aword}$ in chunks of length $h(\card{\aword})$, possibly with the exception of the last chunk, which can be of smaller size. 
            Let $d$ be the number of chunks, and let $\widetilde{\aword}_j$ be the $j$-th chunk analysed 
            by~$\Turing$.
            \begin{enumerate}[label=(\alph*)]
                \setlength{\itemsep}{3pt}
                \item If $\widetilde{\aword}_j \not\in C$, $\Turing$ \textbf{accepts},
                \item if the first chunk $\widetilde{\aword}_1$ is such that $c(\widetilde{\aword}_1) \not\in \Delta(c(\aword^{(i-1)}))$, then $\Turing$ \textbf{accepts}, 
                \item if $\widetilde{\aword}_j$ is not the last chunk and contains a symbol from~$\States_\exists$, then $\Turing$ \textbf{accepts},
                \item if $\widetilde{\aword}_j$ is not the last chunk and $c(\widetilde{\aword}_{j+1}) \not \in \Delta(c(\widetilde{\aword}_j))$, then $\Turing$ \textbf{accepts},
                \item if the length of the last chunk~$\widetilde{\aword}_d$ is not $h(\card{\aword})$ (i.e.~the length of~$\widetilde{\aword}$ is not a multiple of~$h(\card{\aword})$), then $\Turing$ \textbf{accepts},
                \item if the last chunk $\widetilde{\aword}_d$ contains a symbol from~$\States_\forall$, then $\Turing$ \textbf{accepts}.
            \end{enumerate}
            \item\label{turing:step3:substep3} $\Turing$ analyses the last chunk $\widetilde{\aword}_d$ of the previous step. 
            If~$\qacc$ occurs in $\widetilde{\aword}_d$, then $\Turing$ \textbf{accepts}. 
            If~$\qrej$ occurs in~$\widetilde{\aword}_d$ or $i = m$, then $\Turing$ \textbf{rejects}.
            Else,~$\Turing$ writes down $\widetilde{\aword}_d \cdot \trail$ at the beginning of the $i$-th tape, and \textbf{moves to macrostep~$(i+1)$}. Let $\aword^{(i)} = \widetilde{\aword}_d$.
        \end{enumerate}
        
        Following the same arguments of the first macrostep, the following claims are established. 

        \begin{claim}
            \label{claim:turing:macrostep-i-even-input}
            Let $i > 1$ even.
            Macrostep~$i$ only reads at most the first $h(\card{\aword})^2$ characters of the input of the $i$-th tape, and at most the first $h(\card{\aword})$ characters of the $(i-1)$-th tape.
            It does not depend on the content written on the other tapes.
        \end{claim}

        \begin{claim}
            \label{claim:turing:macrostep-i-even}
            Let $i > 1$ even. The macrostep~$i$ runs in polynomial time on~$\card{\aword_1}$,~$k$, $\card{\aword}$ and $\card{\ATuring}$,
            and required polynomially many states to be implemented, w.r.t.~$k$, $\card{\aword}$ and $\card{\ATuring}$. 
        \end{claim}

        \begin{claim}
            \label{claim:turing:macrostep-i-even-semantics}
            Let $i > 1$ even, and let $\widetilde{\aword}$ prefix~of~$\aword_i \cdot \trail$, 
            such that~either $\widetilde{\aword} \cdot \trail$ is a prefix of $\aword_i \cdot \trail$ or $\card{\widetilde{\aword}} = h(\card{\aword})^2$.
            Suppose that~$\Turing$ did not halt on the macrosteps~$1,\dots,i-1$.
            \begin{itemize}
                \item if $\aword^{(i-1)} \cdot \widetilde{\aword}$ encodes an universal hop of $\ATuring$ ending in state $\qrej$, then $\Turing$ rejects,
                \item else, if $i = m$ and $\aword^{(i-1)} \cdot \widetilde{\aword}$ encodes an universal hop of $\ATuring$, ending in a state from~$\States_\exists$, then $\Turing$ rejects,
                \item else, if $i < m$ and $\widetilde{\aword}$ encodes an universal hop of $\ATuring$, ending in a state from~$\States_\exists$,
                then $\Turing$ writes the last configuration of this path (i.e.~$\aword^{(i)}$) at the beginning of the $i$-th tape and moves to macrostep $(i+1)$,
                \item otherwise, $\Turing$ accepts.
            \end{itemize}
        \end{claim}
    \end{description}
    
    This ends the definition of~$\Turing$. Since $\Turing$ performs $m = g(\card{\aword})$ macrosteps, where $g$ is a polynomial, from~\Cref{claim:turing:macrostep-1}, \Cref{claim:turing:macrostep-i-odd} and~\Cref{claim:turing:macrostep-i-even} we conclude that 
    $\Turing$ can be constructed in polynomial time in $n > m$ (number of tapes), $k$, $\card{\ATuring}$ and $\card{\aword}$. Moreover, $\Turing$ runs in polynomial time, as required by property~\eqref{lem:almost-theorem-prop-1}.
    By~\Cref{claim:turing:macrostep-1-input},~\Cref{claim:turing:macrostep-i-odd-input} and~\Cref{claim:turing:macrostep-i-even-input}, $\Turing$ only looks at a prefix of its inputs of length at most $h(\card{\aword})^2$, and its runs are independent on the input of the last $n-m$ tapes. 
    Therefore, both the properties~\eqref{lem:almost-theorem-prop-2} and~\eqref{lem:almost-theorem-prop-2b} are satisfied. 
    
    To conclude the proof, it remains to show that~$\Turing$ satisfies property~\eqref{lem:almost-theorem-prop-3}, which 
    follows directly from the following claim, proved thanks to~\Cref{claim:turing:macrostep-1-semantics},~\Cref{claim:turing:macrostep-i-odd-semantics}
    and~\Cref{claim:turing:macrostep-i-even-semantics}.
    \phantom\qedhere
    \end{proof}

    \begin{claim}
        Let $\ell \in [1,m]$ and 
        $(\pi_1,\dots,\pi_\ell)$ be a $\ell$-uple such that, 
        \begin{itemize}[nosep]
            \item $\pi_1 \cdot {\dots} \cdot \pi_\ell$ is a computation path of~$\ATuring$,
            \item for every $i \in [1,\ell]$ odd, $\pi_i = (c_0^i,\dots,c_{d_i}^i)$ is an existential hop,
            \item for every $i \in [1,\ell]$ even, $\pi_i = (c_0^i,\dots,c_{d_i}^i)$ is an universal hop,
            \item $c_0^1 = (\aword,\qinit,0)$.
        \end{itemize}
        For every $i \in [1,\ell]$, let $\aword_i \in \overline{\Alphabet}^{h(\card{\aword})^2}$ be a word encoding the computation path $\pi_i$.
        Then, 
        \begin{equation}
            \label{LastATMClaim}
            \gamma(c^\ell_{d_\ell}) = \tacc
            \text{ iff }
            Q_{\ell+1}\aword_{\ell+1} \in \overline{\Alphabet}^{h(\card{\aword})^2},\dots,Q_m \aworld_m \in \overline{\Alphabet}^{h(\card{\aword})^2} \text{ : }
            \Turing(\aworld_1,\dots,\aworld_m) \text{ accepts},
        \end{equation}
        where for every $i \in [\ell+1,m]$, $Q_i = \exists$ if $i$ is odd, and otherwise $Q_i = \forall$.
    \end{claim}

    \begin{claimproof}
        \let\oldqedsymbol\qedsymbol
        \renewcommand{\claimqedhere}{\renewcommand\qedsymbol{\textcolor{lipicsGray}{\ensuremath{\vartriangleleft}}\oldqedsymbol}%
        \qedhere%
        \renewcommand\qedsymbol{\textcolor{lipicsGray}{\ensuremath{\blacktriangleleft}}}}
        The proof is by induction on $m-\ell$.
        \begin{description}
            \item[base case: $m = \ell$.] 
            Since $\ATuring$ is $g$-alternation bounded and $m = g(\card{\aword})$, in this case $c^m_{d_m}$ is either $\qacc$ of $\qrej$.
            If $\gamma(c^m_{d_m}) = \tacc$, then the state of the last configuration $c_{d_m}^m$
            is the accepting state~$\qacc$, and from~\Cref{claim:turing:macrostep-1-semantics},~\Cref{claim:turing:macrostep-i-odd-semantics}
            and~\Cref{claim:turing:macrostep-i-even-semantics},
            we conclude that $\ATuring(\aword_1,\dots,\aword_m)$ accepts.
            Else, suppose that $\Turing(\aworld_1,\dots,\aworld_m)$ accepts. 
            Then, since ${\aword_1\cdot{\dots}\cdot\aword_m}$ encodes a computation path of $\ATuring$, by~\Cref{claim:turing:macrostep-1-semantics},~\Cref{claim:turing:macrostep-i-odd-semantics}
            and~\Cref{claim:turing:macrostep-i-even-semantics},
            the state of the last configuration $c_{d_m}^m$
            cannot be~$\qrej$.
            Hence, it is $\qacc$, and $\gamma(c^m_{d_m}) = \tacc$. 


            \item[induction step: $m - \ell > 1$.] 
            First of all, if the path $\pi_1\cdot{\dots}\cdot\pi_\ell$ is terminating, then~\eqref{LastATMClaim} follows exactly as in the base case.
            Below, let us assume that $\pi_1\cdot{\dots}\cdot\pi_\ell$ is not terminating. 
            Notice that this means that the length of the computation path is less than $h(\card{\aword})$, since the \ATM $\ATuring$ is $h$-time bounded.
            Below, $\overline{c}$ is short for $c^\ell_{d_\ell}$ and we write~$\overline{\aword}$ for the word from $C$ that encodes $\overline{c}$.
            We split the proof depending on the parity of~$\ell$.
            
            \begin{description}
                \item[case: $\ell$ even.]
                    In this case, $Q_{\ell+1} = \exists$,
                    $\pi_\ell$ is an universal hop, 
                    and $\overline{c}$ contains a state 
                    in~$\States_\exists$. 
                    \ProofRightarrow
                    Suppose $\gamma(\overline{c}) = \tacc$. Since $\overline{c}$ is an existential configuration, this implies that there a computation path $\pi_{\ell+1} = (c_0,\dots,c_d)$ such that $\overline{c} \cdot \pi_{\ell+1}$ is an existential hop, and $\gamma(c_d) = \tacc$. 
                    Let $\aword_{\ell+1}$ be a word encoding $\pi_{\ell+1}$, of minimal length. Since $\ATuring$ is $h$-time bounded, $\aword_{\ell+1} \in \overline{\Alphabet}^{h(\card{\aword})^2}$.
                    By induction hypothesis,
                    \begin{center}
                    $Q_{\ell+2}\aword_{\ell+2} \in \overline{\Alphabet}^{h(\card{\aword})^2},\dots,Q_m \aworld_m \in \overline{\Alphabet}^{h(\card{\aword})^2}$ : 
                    $\Turing(\aworld_1,\dots,\aworld_m)$ accepts.
                    \end{center}
                    Hence, 
                    \begin{center}
                    $\exists \aword_{\ell+1} \in \overline{\Alphabet}^{h(\card{\aword})^2} Q_{\ell+2}\aword_{\ell+2} \in \overline{\Alphabet}^{h(\card{\aword})^2},\dots,Q_m \aworld_m \in \overline{\Alphabet}^{h(\card{\aword})^2}$ : 
                    $\Turing(\aworld_1,\dots,\aworld_m)$ accepts.
                    \end{center}
                    \ProofLeftarrow
                    Conversely, suppose that the right hand side of~\eqref{LastATMClaim} holds. Since $Q_{\ell+1} = \exists$, from~\Cref{claim:turing:macrostep-i-odd-semantics} 
                    there is a word $\aword_{\ell+1} \in \overline{\Alphabet}^{h(\card{\aword})^2}$ 
                    that encodes a computation path $\pi_{\ell+1} = (c_0,\dots,c_d)$ of $\ATuring$, such that $\overline{c} \cdot \pi_{\ell+1}$ is an existential hop and 
                    \begin{center}
                        $Q_{\ell+2}\aword_{\ell+2} \in \overline{\Alphabet}^{h(\card{\aword})^2},\dots,Q_m \aworld_m \in \overline{\Alphabet}^{h(\card{\aword})^2}$ : 
                        $\Turing(\aworld_1,\dots,\aworld_m)$ accepts.
                    \end{center}
                    By induction hypothesis, $\gamma(c_d) = \tacc$, which yields $\gamma(\overline{c}) = \tacc$, by definition of~$\gamma$.
                \item[case: $\ell$ odd.]
                    In this case, $Q_{\ell+1} = \exists$,
                    $\pi_\ell$ is an existential hop, 
                    and $\overline{c}$ contains a state 
                    in~$\States_\forall$. 
                    \ProofRightarrow
                    Suppose $\gamma(\overline{c}) = \tacc$. Since $\overline{c}$ is an universal configuration, $\gamma(c_d) = \tacc$ holds for every computation path $\pi = (c_0,\dots,c_d)$ such that $c^\ell_{d_\ell} \cdot \pi$ is an universal hop.
                    Let $\aword_{\ell+1}$ be a word in $\overline{\Alphabet}^{h(\card{\aword})^2}$. 
                    If $\overline{\aword} \cdot \aword_{\ell+1}$ does not encode an universal hop, then from~\Cref{claim:turing:macrostep-i-even-semantics}, 
                    $\Turing(\aword_1,\dots,\aword_{\ell+1},x_{\ell+2},\dots,x_m)$ accepts for every values of~$x_{\ell+2},\dots,x_m$.
                    If 
                    Suppose instead that $\overline{\aword} \cdot \aword_{\ell+1}$ encodes the universal hop $\pi_{\ell+1} = (c_0,\dots,c_d)$. 
                    From $\gamma(c_d) = \tacc$ and by induction hypothesis,
                    \begin{center}
                    $Q_{\ell+2}\aword_{\ell+2} \in \overline{\Alphabet}^{h(\card{\aword})^2},\dots,Q_m \aworld_m \in \overline{\Alphabet}^{h(\card{\aword})^2}$ : 
                    $\Turing(\aworld_1,\dots,\aworld_m)$ accepts.
                    \end{center}
                    We conclude that
                    \begin{center}
                    $\forall \aword_{\ell+1} \in \overline{\Alphabet}^{h(\card{\aword})^2} Q_{\ell+2}\aword_{\ell+2} \in \overline{\Alphabet}^{h(\card{\aword})^2},\dots,Q_m \aworld_m \in \overline{\Alphabet}^{h(\card{\aword})^2}$ : 
                    $\Turing(\aworld_1,\dots,\aworld_m)$ accepts.
                    \end{center}
                    \ProofLeftarrow
                    Conversely, suppose that the right hand side of~\eqref{LastATMClaim} holds. Since $Q_{\ell+1} = \forall$, this means that for every $\aword_{\ell+1} \in \overline{\Alphabet}^{h(\card{\aword})^2}$ 
                    \begin{center}
                        $Q_{\ell+2}\aword_{\ell+2} \in \overline{\Alphabet}^{h(\card{\aword})^2},\dots,Q_m \aworld_m \in \overline{\Alphabet}^{h(\card{\aword})^2}$ : 
                        $\Turing(\aworld_1,\dots,\aworld_m)$ accepts.
                    \end{center}
                    Since $\overline{c}$ is an universal configuration, in order to conclude that 
                    $\gamma(\overline{c}) = \tacc$ holds, it is sufficient to check that, for every computation path $\pi_{\ell+1} = (c_0,\dots,c_d)$ of $\ATuring$, 
                    if $\overline{c} \cdot \pi_{\ell+1}$ is an universal hop, then $\gamma(c_d) = \tacc$.
                    So, consider a computation path $\pi_{\ell+1} = (c_0,\dots,c_d)$ of $\ATuring$, 
                    such that $\overline{c} \cdot \pi_{\ell+1}$
                    is an universal hop. 
                    Since $\ATuring$ is $h$-time bounded, 
                    there is $\aword_{\ell+1} \in \overline{\Alphabet}^{h(\card{\aword})^2}$ that encodes $\pi_{\ell+1}$. We have, 
                    \begin{center}
                        $Q_{\ell+2}\aword_{\ell+2} \in \overline{\Alphabet}^{h(\card{\aword})^2},\dots,Q_m \aworld_m \in \overline{\Alphabet}^{h(\card{\aword})^2}$ : 
                        $\Turing(\aworld_1,\dots,\aworld_m)$ accepts,
                    \end{center}
                    and by induction hypothesis we obtain $\gamma(c_d) = \tacc$, concluding the proof. 
                    \claimqedhere
            \end{description}
        \end{description}
    \end{claimproof}

The proof of~\Cref{theorem:kaexp-alternation-complexity} stems directly from~\Cref{lemma:almost-theorem-kaexp-alternation}.

\TheoremKAEXPAltProblem*

\begin{proof}
    Fix~$k \geq 1$.
    Since the $k\aexppol$-prenex problem is solvable in $k\aexppol$ directly from its definition,
    we focus on the $k\aexppol$-hardness.
    We aim at a reduction from the membership problem described in~\Cref{proposition:kaexppol-membership}, 
    relying on~\Cref{lemma:almost-theorem-kaexp-alternation}.
    Let~$f,g \colon \Nat \to \Nat$ be two polynomials, and define~$h(\avar) \egdef \tetra(k,f(\avar))$.
    Consider a~$h$-time bounded, $g$-alternation bounded~\ATM~$\ATuring = (\Alphabet,\States_{\exists},\States_{\forall},\qinit,\qacc,\qrej,\trans)$, as well as a word~$\aword \in \Alphabet^*$. 
    We want to check whether $\aword \in \alang(\ATuring)$.
    To simplify the proof without loss of generality, 
    we assume $\card{\aword} \geq 1$ and that~$f(\avar) \geq \avar$ for every $\avar \in \Nat$. 
    Then,~$f(\card{\aword}) \geq \card{\aword} \geq 1$.
    Let $m = \max(f(\card{\aword}),g(\card{\aword}))+3$, 
    and let~$\overline{\Alphabet} \egdef \Alphabet \cup \States \cup \{\qacc,\qrej,\trail\}$.
    We consider the $m$DTM~${\Turing = (m,\overline{\Alphabet},\States',\qinit',\qacc',\qrej',\trans')}$ derived from~$\ATuring$ by applying~\Cref{lemma:almost-theorem-kaexp-alternation}.
    From the property~\eqref{lem:almost-theorem-prop-1} of~$\Turing$, there is a polynomial~$p \colon \Nat \to \Nat$ such that, for every~$(\aword_1,\dots,\aword_m) \in (\overline{\Alphabet}^*)^m$,~$\Turing(\aworld_1,\dots,\aworld_m)$ halts in time $p(\max(\card{\aword_1},\dots,\card{\aword_m}))$. W.l.o.g., we can assume $p(\avar) = \alpha \avar^d + \beta$ for some~${\alpha,d,\beta \in \PNat}$.
    So, for all~$\avar \in \Nat$, 
    $p(\avar) \geq \avar$.
    We rely on~$\Turing$ to build a polynomial instance of 
    the~$k\aexppol$-prenex problem that is satisfied if and only if $\aword \in \alang(\ATuring)$.
    Let~${n \egdef p(2 \cdot m)}$, and consider
    the \nTM~${\Turing' = (n,\overline{\Alphabet},\States',\qinit',\qacc',\qrej',\trans')}$. 
    Notice that $\Turing'$ is defined exactly as~$\Turing$, but has $n \geq m$ tapes, where $n$ is polynomial in~$\card{\aword}$. 
    However, since the two machines share the same transition function~$\trans'$, only the first $m$ tapes are effectively used by~$\Turing'$, and for all~$({\aworld_1,\dots,\aworld_m,\dots,\aworld_n) \in \overline{\Alphabet}}^n$, 
    $\Turing(\aworld_1,\dots,\aworld_m)$ and $\Turing'(\aworld_1,\dots,\aworld_n)$ perform the same computational steps.
    Let $\vec{Q} = (Q_1,\dots,Q_n) \in \{\exists,\forall\}^n$ such that for all $i \in [1,n]$, $Q_i = \exists$ iff $i$~is~odd. We consider the instance of the 
    $k\aexppol$-prenex problem given by~$(n,\vec{Q},\Turing')$. 
    This instance is polynomial in~$\card{\aword}$ and~$\card{\ATuring}$, and asks whether 
    \begin{equation}
        \label{proofT8:eq1}
        Q_1 \aworld_1 \in \overline{\Alphabet}^{\tetra(k,n)}, \dots, Q_n \aworld_n \in \overline{\Alphabet}^{\tetra(k,n)} \text{ : }
        \Turing'(\aworld_1,\dots,\aworld_n) \text{ accepts in time }\tetra(k,n).
    \end{equation}
    To conclude the proof, we show that~\Cref{proofT8:eq1} holds if and only if 
    \begin{equation}
        \label{proofT8:eq2}
        Q_1 \aworld_1 \in \overline{\Alphabet}^{h(\card{\aword})^2},\dots, Q_m \aworld_m \in \overline{\Alphabet}^{h(\card{\aword})^2} \text{ : } 
            \Turing(\aworld_1,\dots,\aworld_m) \text{ accepts}.
    \end{equation}
    Indeed, directly from the property~\eqref{lem:almost-theorem-prop-3} of~$\Turing$,
    this equivalence implies that \Cref{proofT8:eq1} is satisfied if and only if 
    $\aword \in \alang(\ATuring)$, 
    leading to 
    the $k\aexppol$-hardness of the~$k\aexppol$-prenex problem
    directly by~\Cref{proposition:kaexppol-membership}.
    
    By definition of the polynomial $p \colon \Nat \to \Nat$, 
    if~$\Turing(\aworld_1,\dots,\aworld_m)$ accepts, then it does so in at most $p(\max(\card{\aword_1},\dots,\card{\aword_m}))$ steps. Therefore,~\Cref{proofT8:eq2} holds if and only if 
    \begin{equation}
        \label{proofT8:eq3}
        Q_1 \aworld_1 \in \overline{\Alphabet}^{h(\card{\aword})^2},\dots, Q_m \aworld_m \in \overline{\Alphabet}^{h(\card{\aword})^2} \text{ : } 
            \Turing(\aworld_1,\dots,\aworld_m) \text{ accepts in time}~p(h(\card{w})^2).
    \end{equation}
    By~\Cref{lemma:tetra-property-1,lemma:tetra-property-2}, we have
    \begin{equation}
        \label{proofT8:eq4}
    \tetra(k,n) \ = \ \tetra(k,p(2 \cdot m)) \ \geq \ p(\tetra(k,2 \cdot f(\card{\aword}))) \ \geq \ p(\tetra(k,f(\card{\aword})^2)) \ \geq \  p(h(\card{\aword})^2).
    \end{equation}
    This chain of inequalities implies 
    that~\Cref{proofT8:eq3} holds if and only if 
    \begin{equation}
        \label{proofT8:eq5}
        Q_1 \aworld_1 \in \overline{\Alphabet}^{h(\card{\aword})^2},\dots, Q_m \aworld_m \in \overline{\Alphabet}^{h(\card{\aword})^2} \text{ : } 
            \Turing(\aworld_1,\dots,\aworld_m) \text{ accepts in time}~\tetra(k,n),
    \end{equation}
    where, again, the right to left direction holds since, if $\Turing(\aworld_1,\dots,\aworld_m)$ accepts, then it does so in at most~$p(h(\card{w})^2)$ steps.
    Now, from the property~\eqref{lem:almost-theorem-prop-2} of $\Turing$, and again by~\Cref{proofT8:eq4} together with the assumption $p(\avar) \geq \avar$ (for all~$\avar \in \Nat$), 
    we can extend the bounds on the words $\aword_1,\dots,\aword_m$ we consider, and
    conclude that~\Cref{proofT8:eq5} is equivalent to
    \begin{equation}
        \label{proofT8:eq6}
        Q_1 \aworld_1 \in \overline{\Alphabet}^{\tetra(k,n)},\dots, Q_m \aworld_m \in \overline{\Alphabet}^{\tetra(k,n)} \text{ : } 
            \Turing(\aworld_1,\dots,\aworld_m) \text{ accepts in time}~\tetra(k,n).
    \end{equation}
    Lastly, since for all~$({\aworld_1,\dots,\aworld_m,\dots,\aworld_n) \in \overline{\Alphabet}}^n$, 
    $\Turing(\aworld_1,\dots,\aworld_m)$ and $\Turing'(\aworld_1,\dots,\aworld_n)$ perform the same computational steps,
    we conclude that~\Cref{proofT8:eq6} holds iff~\Cref{proofT8:eq1} holds.
\end{proof}

\CorollarySIGMAKProblem* 

\begin{proof}(\textit{sketch})
    Let~$f \colon \Nat \to \Nat$ be a polynomials, 
    and define~$h(\avar) \egdef \tetra(k,f(\avar))$.
    Consider a~$h$-time $j$-alternation bounded~\ATM~$\ATuring$ (i.e.~$\ATuring$ alternates between existential and universal states at most $j \in \PNat$ times), as well as a word~$\aword \in \Alphabet^*$. 
    We want to check whether $\aword \in \alang(\ATuring)$.
    Let $m = \max(f(\card{\aword}),j)+3$.
    One can easily 
    revisit~\Cref{lemma:almost-theorem-kaexp-alternation} with respect to~$\ATuring$ as above,
    so that the property~\eqref{lem:almost-theorem-prop-3} of the $m\TM$ $\Turing$ is updated as follows.
    \begin{center}
    $\aword \in \alang(\ATuring)$ \ iff \
    $Q_1 \aworld_1 \in \overline{\Alphabet}^{h(\card{\aword})^2},\dots, Q_m \aworld_m \in \overline{\Alphabet}^{h(\card{\aword})^2}$ : 
    $\Turing(\aworld_1,\dots,\aworld_m)$ accepts,
    \end{center}
where $Q_1 = \exists$, for every $i \in [1,j]$, $Q_{i} \neq Q_{i+1}$, and for every $i \in [j+1,m]$, $Q_{i} = Q_{j+1}$. In particular, now $\vec{Q}= (Q_1,\dots,Q_m)$ is such that $\altern{\vec{Q}} = j$.
The proof then carries out analogously to the proof of~\Cref{theorem:kaexp-alternation-complexity}, the only difference being that for the \nTM $\Turing'$, we consider the quantifier prefix $(Q_1,\dots,Q_n)$ such that, for every $i \in [m+1,n]$, $Q_i = Q_m$, in order to keep the number of alternations bounded by $j$.
\end{proof}

\bibliography{bibliography}

\appendix

\end{document}